\documentclass[a4paper,11pt]{article}

\usepackage{amssymb,amsfonts,amsthm,amsmath,dsfont, graphicx}
\usepackage[cp1251]{inputenc}
\usepackage[english]{babel}
\usepackage{mathrsfs}
\usepackage{mathtext}
\usepackage{verbatim}
\usepackage{caption}
\usepackage{subcaption}
\usepackage{listings}
\usepackage{color, multirow}
\usepackage{booktabs}
\usepackage{authblk}

\usepackage[authoryear]{natbib}
\newtheorem{lemma}{Lemma}
\newtheorem{Th}{Theorem}

\newtheorem{remark}{Remark}
\newtheorem{Cor}{Corollary}

\newtheorem*{ass}{Assumption}

\usepackage{color}

\usepackage{ulem}

\newcommand{\T}{^\prime}
\newcommand{\VGMV}{V_{GMV}}
\newcommand{\RGMV}{R_{GMV}}
\newcommand{\1}{\mathbf{1}}
\newcommand{\w}{\boldsymbol{\omega}}
\newcommand{\bmu}{\boldsymbol{\mu}}
\newcommand{\bSigma}{\boldsymbol{\Sigma}}
\newcommand{\bR}{\boldsymbol{R}}
\newcommand{\D}{\mathcal{D}}
\newcommand{\inv}{^{-1}}

\begin{document}

\title{
{Mean-Variance Efficiency of Optimal Power and Logarithmic Utility Portfolios}
}

\date{}
\author[a]{{\small Taras Bodnar}}
\author[b]{{\small Dmytro Ivasiuk}}
\author[c, \footnote{Corresponding Author: Nestor Parolya. E-Mail: n.parolya@tudelft.nl.}]{{\small Nestor Parolya}}
\author[b]{{\small Wolfgang Schmid}}

\affil[a]{{\footnotesize Department of Mathematics, Stockholm University, Stockholm, Sweden}}
\affil[b]{{\footnotesize Department of Statistics, European University Viadrina, Frankfurt(Oder), Germany}}
\affil[c]{{\footnotesize Delft Institute of Applied Mathematics, Delft University of Technology, Netherlands}}
\date{}
\maketitle

\begin{abstract}		
We derive new results related to the portfolio choice problem for power and logarithmic utilities. Assuming that the portfolio returns follow an approximate log-normal distribution,  the closed-form expressions of the optimal portfolio weights are obtained for both utility functions. Moreover, we prove that both optimal portfolios belong to the set of mean-variance feasible portfolios and establish necessary and sufficient conditions such that they are mean-variance efficient. Furthermore, an application to the stock market is presented and the behavior of the optimal portfolio is discussed for different values of the relative risk aversion coefficient. It turns out that the assumption of log-normality does not seem to be a strong restriction.
\end{abstract}


\section{Introduction}

The theory of optimal portfolio choice started with the pioneering contribution of \cite{markowitz1952portfolio}. Markowitz used the variance as a measure of the risk of a  portfolio return. He recommended to choose the portfolio weights in such a way that the portfolio variance is minimal for a given  level of the expected portfolio return. All of these so-called efficient portfolios lie in the efficient frontier which is a parabola in the mean-variance space.

In the meantime many further proposals for a portfolio selection have been made.  A widely made approach is based on the maximization of an investor's utility function,  where the investor chooses a portfolio for which its utility has reaches a maximum possible value \citep{pennacchi2008theory, elton2009modern}. The mean-variance approach of \cite{merton1974fallacy} turns out to be fully consistent with the expected utility maximization \citep[see,][]{dybvig1982mean}. This result is valid without any distributional assumptions imposed on the returns.  Similarly, a quadratic utility provides  a closed-form solution under very general conditions \citep{bodnar2015closed}. However, there are many other ways to choose the utility function like, e.g., the power and the exponential utility function. In these cases no closed-from solutions can be derived without information on the distribution of the return process \citep{bodnar2015exact}.

The focus of this paper lies on the power and on the logarithmic utility functions. If $W$ denotes the wealth of the investor, then the power utility is given by $U\left(W\right)=\frac{1}{1-\gamma}W^{1-\gamma},\ \gamma>0, \gamma \neq 1$.
The logarithmic utility $U\left(W\right)=\log W$ is obtained as a limit of the power utility if $\gamma$ converges to one. The quantity $\gamma$ is equal to the relative risk aversion which is constant for the power utility (CRRA). This means that an investor at any wealth level will hold the same amount of money on the risky assets \citep{brandt2009portfolio,campbell2002strategic}. This is a very important property and it makes this utility function quite attractive from economical point of view.

In order to determine the optimal portfolio for the power utility several numerical procedures have been proposed in literature \citet{grauer1987gains,cover1991universal}.
The first-order conditions for the maximum of the expected utility are derived and a Taylor series expansion is used to transform the expected utility into a linear combination of moments \citep{brandt2005simulation}. However, no analytical solution is available in literature and mostly a numerical approach is applied \citep{brandt2009portfolio}. Yet another possibility is to make use of distributional assumptions on the returns predictability \citep{barberis2000investing} and to combine it with a numerical procedure \citep{brandt2006dynamic, lynchtan2010}.

The use of the log-normal distribution for modeling asset or portfolio returns has become popular many years ago. In \citet{elton1974portfolio} an efficient set theorem is derived assuming log-normally distributed portfolio returns. They present an algorithm for determining the efficient portfolio. In the paper on multi-period mean variance analysis \citet{hakansson1971multi} describes a limiting approach for the optimal portfolio choice. \citet{merton1974fallacy} discuss the log-normal approach in detail and conclude that it should be very carefully applied.

In this paper we  consider the problem of finding the optimal  portfolios for a power and logarithmic utilities under the assumption that the portfolio gross returns are approximately log-normally distributed. The analytical expressions of the optimal portfolio weights are obtained for both utility functions. Moreover, we state the conditions under which the optimal portfolio in the sense of maximizing the expected power utility and the optimal portfolio that maximizes the expected logarithmic utility are both mean-variance efficient.

The log-normal distribution used  as a model for portfolio returns in this paper has several nice properties.  If the portfolio gross return follows a log-normal distribution, then it is possible to determine its higher-order moments in a straightforward way. Moreover, if the ratio of mean and standard deviation is large (analog of Sharpe ratio), then the log-normal distribution is close to the normal one (analytically shown in the paper). This is for instance the case when the portfolio mean is bounded and the portfolio variance is small.

In applications, the asset returns $r_j$ vary in most cases in a range up to $\pm 15\%$. Thus the corresponding gross returns $R_j = 1 + r_j$ take values close to $1$. Consequently the overall portfolio mean is close to one while the portfolio variance is usually small. This means that in such a case it behaves like a normal distribution. Moreover, assuming a log-normally distributed  portfolio return, its continuously compounded rate of return  appears to be normally distributed and vice versa. It is also notable that the normality assumption is often used in a literature, i.e. to model the asset return path, VaR process with normally distributed error terms is taken (\citet{bodnar2015closed,bodnar2015exact,brandt2005simulation,brandt2006dynamic,brandt2009portfolio}).

The rest of the paper is structured as follows. In Sections 2 and 3 the optimization problem is described and the assumptions used in the paper are presented and discussed. Our main findings are given in Section 4. In Theorem 1 we derive an explicit formula for the weights of the optimal portfolio that maximizes the expected power utility and discuss its properties.  In particular, we derived the conditions under which this portfolio is mean-variance efficient in Theorem 2 and show that it converges to the portfolio with the maximum Sharpe ratio when the relative risk aversion converges to infinity. In Section 4.2, the analytical results for the optimal portfolio in the sense of maximizing the expected logarithmic utility are present: the closed-form expression of its weights is provided in Theorem 3, while its mean-variance efficiency is proved in Theorem 4. In Section 5 the results of an empirical study are shown. Here we study portfolios based on stocks contained in the German stock market index. We check whether the log-normality assumption on the wealth is fulfilled and whether the conditions imposed in Theorem 1 are fulfilled. It turns out that these assumptions are reasonable and that they are mostly satisfied for the  portfolios in our study. Moreover, we compare the portfolio with optimal weights with other portfolio strategies and show the dominance of the new approach.

\section{Motivation of the log-normal approximation}
	
Let us define two random variables $X\sim N(\mu,\sigma^2)$ and $Y\sim\ln N(\ln\mu,\frac{\sigma^2}{\mu^2})$.
It is easy to see that the first two moments of the discussed distributions conside if $\sigma\slash\mu$ is small, i.e.,  $E(Y)=e^{\ln\mu+\frac{\sigma^2}{2\mu^2}}\approx\mu$ and $Var(Y)=e^{2\ln\mu+\frac{\sigma^2}{\mu^2}}(e^{\frac{\sigma^2}{\mu^2}}-1)\approx\sigma^2$ as $\sigma\slash\mu\rightarrow0$. Indeed, next we show that if $\sigma\slash\mu\rightarrow0$ and $\mu>0$ then the distribution of $Y$ approaches to the distribution of $X$.

In order to find the accuracy of the approximation we consider the difference of both distribution functions, namely
\begin{equation}\label{psi}
\begin{split}
\Psi(x)&=\Phi\left(\frac{x-\mu}{\sigma}\right)-\Phi\left(\frac{\ln x-\ln\mu}{\sigma\slash\mu}\right)\\
\end{split}.
\end{equation}

\begin{lemma}[Log-normal approximation]\label{lem0}
  Let $\mu>0$ then for $\Psi(x)$ defined in \eqref{psi} holds
  \begin{equation}
  \sup_{x\in \mathbb{R}}|\Psi(x)|=O\left(\frac{\sigma}{\mu} \right)~~\text{as ~$\sigma/\mu\to0$.}    
  \end{equation}
\end{lemma}
Thus, difference between normal and log-normal distribution functions approaches zero quite fast as $\sigma\slash\mu\rightarrow 0$. So, the smaller $\sigma\slash\mu$ the better approximation we receive.

Taking into account that the maximization problem deals with a portfolio gross return, it is well-known that in practice asset returns are close to zero, so one expects to have a portfolio gross return fluctuating around one. Consequently, considering the asset returns to be normally distributed with a mean nearby zero and a small variance we obtain normally distributed gross portfolio returns with a mean around one and a small variance. Meanwhile, the power utility is defined on a nonnegative domain, thus, it is naturally to assume that the expected portfolio gross return is positive and portfolio standard deviation is small. That is why, as it was shown, log-normal distribution can be a good proxy for the normally distributed portfolio returns.
	
As an illustrative example Figure \ref{fig: normal and lognormal density} describes an identical behavior of density functions of normal and log-normal distributions for several combinations of $\mu$ and $\sigma$.   

\begin{figure}
	\includegraphics[width=0.5\textwidth]{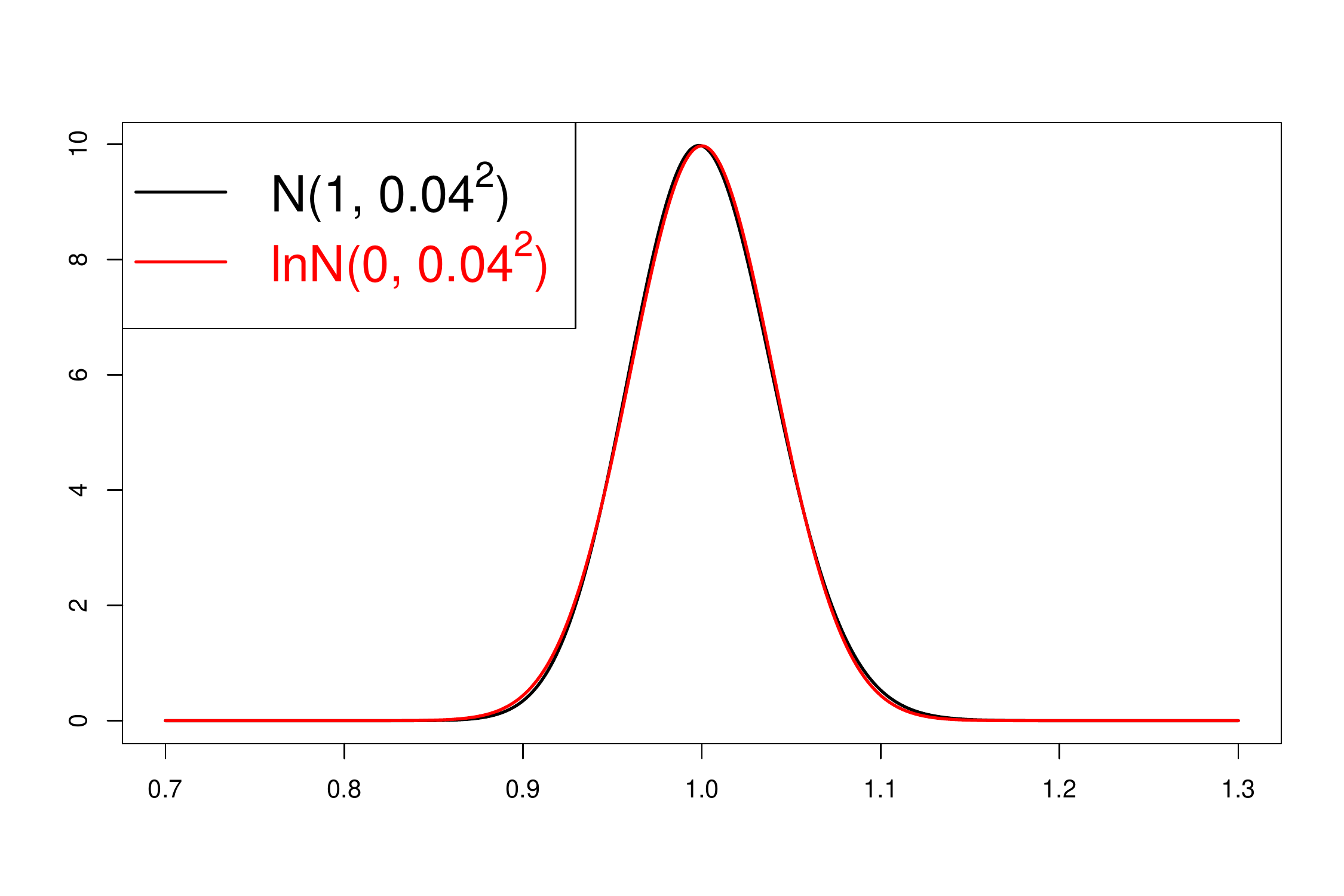} \includegraphics[width=0.5\textwidth]{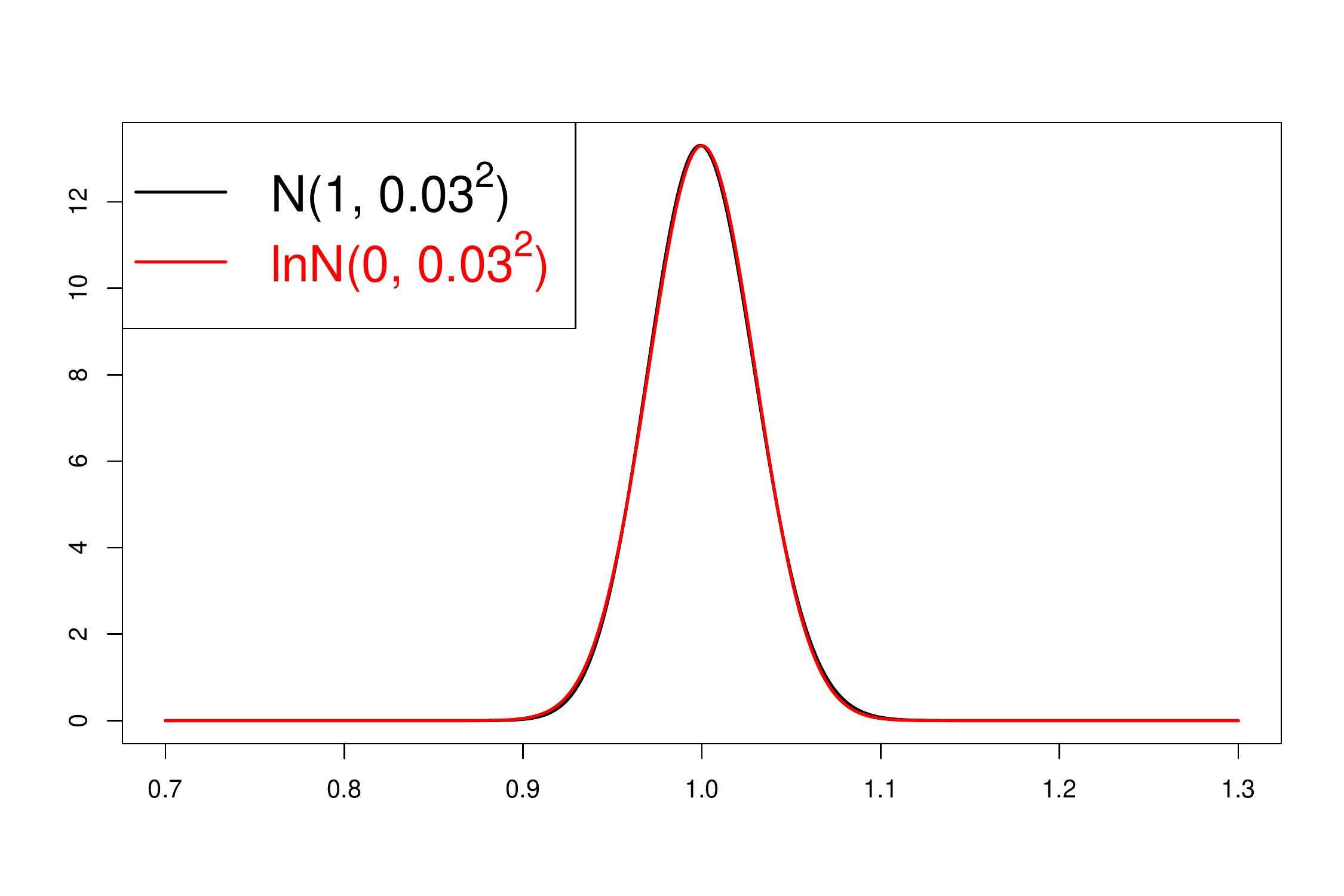}
	\caption{Similarity of normal and log-normal densities.}
	\label{fig: normal and lognormal density}
\end{figure}

\section{Portfolio selection using a utility function}

In this section we give some notations and basic assumptions used in the paper and discuss the portfolio selection problem for a power utility function.

Suppose that a portfolio  consists of $k$ risky assets.  Let $r_i$ denote the relative price change of the $i$th asset within one period  and let $\boldsymbol{r} = (r_1,\dots,r_k)\T $  {be the vector of returns}. The quantity $\omega_i$ stands for the relative amount of money invested into the $i$th asset. Then it holds for the vector of the portfolio weights $\w=\left({\omega}_{1},\omega_{2},\dots,{\omega}_{k}\right)\T $ that $\w\T \1=1$, where $\1$ is a vector of ones. The portfolio return at the end of the period is equal to $\w\T  \boldsymbol{r}$. If $W_0$ denotes the wealth of the investor at the beginning of the period, then   {the} wealth at the end of the period is given by
\begin{equation*}
W=W_{0}\left(\w\T \boldsymbol{r}+1\right)=W_0\w\T \bR,
\end{equation*}
where $\bR=\boldsymbol{r}+\1$.   The mean vector of $\bR$ and its covariance matrix are denoted by $E(\bR)=\boldsymbol{\bmu}$ and $Var(\bR)=\boldsymbol{\bSigma}$ where $\boldsymbol{\bSigma}$ is assumed to be positive definite.

Many  procedures have been proposed in literature how to construct an optimal portfolio, i.e., how to choose the optimal portfolio weights. \cite{markowitz1952portfolio} developed the mean-variance principle for portfolio selection. The idea is to choose the portfolio weights such that the portfolio variance is minimal for a given value of the portfolio return. Another attempt is based on utility functions. It is described in detail by, e.g., \citet{dybvig1982mean}. The optimal weights are obtained by maximizing the expected utility $U\left(\cdot\right)$ of the wealth $W$ at the end of the investment period and it is given by
\begin{equation}\label{eq:investor's problem}
 \w^*=\underset{{\{ \w:\ \w\T \1=1 \}}}{\arg \max} E\left[U\left(W\right)\right].
\end{equation}
There are many possibilities for the choice of a utility function  $U(\cdot)$ \cite[see,][]{pennacchi2008theory}. In the case of the quadratic utility the solution of \eqref{eq:investor's problem} is mean-variance efficient, i.e., it coincides with Markowitz's optimization problem under some conditions \cite[see,][]{bodnar2013equivalence}.  Other utility functions considered in the literature are the HARA, the CARA, and  {the} SAHARA utilities \cite[see, e.g.,][]{pennacchi2008theory,chen2011modeling}.

In this paper we   {deal with the optimal portfolio choice problem under} the power utility function  {expressed as}
\begin{equation}\label{eq: power utility}
U(W)=\frac{W^{1-\gamma}}{1-\gamma},\quad\gamma>0, \gamma \neq 1 .
\end{equation}
It is the only utility function with constant relative risk aversion \cite[cf.][]{pennacchi2008theory}  equal to $\gamma$.   {The} logarithmic utility is a limiting case of the power utility   {when} $\gamma$ converges to $1$  {and it is given by}
\begin{equation}\label{eq: log utility}
U(W)=\log W.
\end{equation}
 Both the power utility function and the logarithmic utility function are widely applied in practice.

In order to calculate the expected utility it is necessary to have a statement on the distribution of the portfolio return. There have been many proposals in literature how to model the return of stocks. \cite{bachelier1900theorie} proposed to use the normal distribution for asset returns while \cite{fama1965portfolio}  and \cite{mittnik1993modeling}  {suggested} the use of stable distributions. An overview on various proposals can be found in, e.g., \citet[][]{jondeau2007financial}.

Here we choose a slightly different way. In light of the discussion in Section 2 we assume that the portfolio return of the investor at the end of the investment period can be well approximated by a log-normal distribution,  which is characterized by two parameters $\alpha \in I\!\!R$ and $\beta > 0$.
We briefly write $Z \sim \ln\mathcal{N}\left(\alpha,\beta^2\right)$ to denote that a random variable $Z$ is log-normally distributed. It holds that
$Z \sim \ln\mathcal{N}\left(\alpha,\beta^2\right)$   {if and only if} $\ln \, Z \sim \mathcal{N}\left(\alpha,\beta^2\right)$   {with the} density  given by
\[ f(z; \alpha, \beta) = \frac{1}{z} \; \frac{1}{\sqrt{2\pi} \beta} \; exp( - \frac{(ln \, z - \alpha)^2}{2 \beta^2} ) \; , \quad z > 0,  \]
and $0$ otherwise. This distribution has a nice property that its moments can be easily derived by
\begin{equation}\label{eq: moments of log-normal distribution}
E\left(Z^\tau\right)=\exp\left(\alpha \tau+\frac{1}{2}\beta^2\tau^2\right), \quad \tau\in\mathbb{R} .
\end{equation}

The log-normal distribution has been applied in finance by several authors (see, e.g., \citet{mcdonald1996}, \citet{LimpertStahel2001}). If  the portfolio return is assumed to be log-normally distributed, then its continuously compounded rate of return is normally distributed and vice versa.  Moreover, doing so it is implicitly assumed that the wealth is positive.

 {\section{Closed-form solutions of optimal portfolio choice problems}	
In this section we present the analytical expressions of the optimal portfolio weights obtained for both the power utility and the logarithmic utility.
Under the assumption that the portfolio gross return $\w\T  \bR$ at the end of the investment period is approximately log-normally distributed, namely $\w\T  \bR \sim \ln\mathcal{N}\left(\alpha,\beta^2\right)$, we get (see Lemma \ref{lem1} in the Appendix) that
\begin{eqnarray*}
\alpha  =  \ln \, X^2 - \frac{1}{2} \ln \, (V + X^2 ) ~~ \text{and} ~~
\beta^2 =  \ln \, \left(1 + \frac{V}{X^2}\right) ,
\end{eqnarray*}
where
\begin{equation}\label{XandV}
X=E(\w\T \bR)=\w\T  \boldsymbol{\bmu}~~
\text{and} ~~ V=Var(\w\T \bR)=\w\T \boldsymbol{\bSigma}\w
\end{equation}
are the expected return and the variance of the portfolio with the weights $\w$. Now we are ready to formulate our main assumption, which justifies a good log-normal approximation.
\begin{ass}\label{A1}
$\mathbf{(A1)}$~~~For any $\w$ it holds $\frac{\sqrt{V}}{X}\to0$.
\end{ass}
For example, if the portfolio gross return is bounded and the asset universe is large, i.e., $k\to\infty$, then it is easy to check that many classical portfolios, like global minimum variance portfolio and equally weightes portfolio satisfy the assumption $\mathbf{(A1)}$ if the covariance matrix of the asset returns $\bSigma$ is well behaved for increasing dimension $k$ (its largest eigenvalue is uniformly bounded in $k$). Thus, one can expect that the log-normal approximation works better in case of large dimensional portfolios.

In the following we also use of the set of efficient constants, i.e. of three quantities which uniquely determine the location of the efficient frontier, the set of optimal portfolios obtained as solutions of Markowitz's optimization problem, in the mean-variance space. They are given by \citep{bodnar2009econometrical}
\begin{equation}\label{eq: eff set const: R, V, s}
\RGMV   =  \frac{\1\T \bSigma\inv\bmu}{\1\T \bSigma\inv\1}, \quad \VGMV =\frac{1}{\1\T \bSigma\inv\1} , \quad \text{and} \quad s = \bmu\T \boldsymbol{Q} \bmu
\end{equation}
with
\begin{equation}\label{eq:Q}
  \boldsymbol{Q} = \bSigma\inv - \frac{\bSigma\inv \boldsymbol{1} \boldsymbol{1}\T  \bSigma\inv}{\boldsymbol{1}\T  \bSigma\inv \boldsymbol{1} }.
\end{equation}
Here, $\RGMV $ denotes the expected portfolio return of a global minimum variance (GMV) portfolio (the portfolio with the smallest variane), $\VGMV $ stands for its variance, and $s$ represents the slope parameter of the efficient frontier which is the upper part of the parabola given by
\begin{equation}\label{eq: Mean-Variance parabola}
(X-\RGMV )^2=s(V-\VGMV )\,.
\end{equation}
The whole parabola is known in the literature as a feasible set of optimal portfolios.

\subsection{Analytical solution for the power utility}

In Theorem \ref{th: optimal portfolio weights theorem} we present the closed-form solution of the optimal portfolio choice problem for the power utility function, i.e. the solution of \eqref{eq:investor's problem} with $U(\cdot)$ given in \eqref{eq: power utility}.

\begin{Th}\label{th: optimal portfolio weights theorem}
  Assume that $\mathbf{(A1)}$ holds. If
  {\footnotesize 
  \begin{eqnarray}\label{gamma_min}
    \gamma\geq\gamma_{min}=2s+2\left(s(1+s)\frac{V_{GMV}}{R^2_{GMV}}+\sqrt{s(1+s)\left(1+\frac{sV_{GMV}}{R^2_{GMV}}\right)\left(1+(1+s)\frac{V_{GMV}}{R^2_{GMV}}\right)}\right),
    \end{eqnarray}}
  then the solution of the optimization problem \eqref{eq:investor's problem} with the power utility function $U(\cdot)$ from \eqref{eq: power utility} is given by
\begin{equation}\label{eq:optimal portfolio weights}
	\w^*=\bSigma \inv\left[\left(-\1+(\gamma+1)\frac{\bmu}{X}\right)\frac{Y}{\gamma}-X\bmu\right]
      \end{equation}
with
  \begin{eqnarray}\label{eq:X}
	X & = & \frac{(\gamma+2)\RGMV  - \sqrt{\D}}{2(1+s)}, \\
    Y & = & \frac{\gamma}{s}(X \RGMV  -\RGMV ^2-s\VGMV ),\label{eq:Y}\\
     \D&=&  (\gamma+2)^2\RGMV ^2-4(\gamma+1)(1+s)(\RGMV ^2+s\VGMV )   \label{ineq: condition of theorem: optimal portfolio, power utility }\,.
        \end{eqnarray}
        Moreover, if $\gamma=\gamma_{min}$ then $\D=0$ and \eqref{eq:X} simplifies to
        \begin{eqnarray*}
 	X & = & \frac{(\gamma_{min}+2)\RGMV}{2(1+s)}.
        \end{eqnarray*}
      %
\end{Th}	

The proof of Theorem \ref{th: optimal portfolio weights theorem} is given in the appendix.
\begin{remark}\rm
  Theorem \ref{th: optimal portfolio weights theorem} reveals one very important fact: the solution for the power utility function not always exists. Indeed, to guarantee the existence we need to have $\gamma\geq \gamma_{min}$, which implies that there exists a minimum level of risk aversion $\gamma_{min}>0$ so that for all $\gamma\geq\gamma_{min}$ the solution of the optimization problem \eqref{eq:investor's problem} with the power utility function exists and is unique. Going through the proof of Theorem \ref{th: optimal portfolio weights theorem} one can justify that this condition is equivalent to a more technical one, i.e., $\D\geq0$ for all $\gamma>0$.
\end{remark}

The expected return of the optimal portfolio that maximizes the expected power utility is $X$, while its variance is equal to $V=Y-X^2$. Moreover, the maximum value of the optimization problem \eqref{eq:investor's problem} is given by
\begin{equation}\label{eq: max expected utility}
\max_{\{ \w:\ \w\T \1=1 \}} E\left[U(W )\right]=\frac{W^{1-\gamma}_{0}}{1-\gamma}\exp\left[(1-\gamma^2)\ln X +\frac{1}{2}(\gamma^2-\gamma)\ln Y\right].
\end{equation}

It is remarkable that the optimal portfolio in the sense of maximizing the expected power utility function is located for any relative risk aversion coefficient $\gamma$ on the set of feasible portfolios, i.e., on the parabola \eqref{eq: Mean-Variance parabola}. This result is formulated as Corollary \ref{cor: parabola} with the proof provided in the Appendix.

\begin{Cor}\label{cor: parabola}
Under the assumptions of Theorem \ref{th: optimal portfolio weights theorem}, it holds that
	\begin{equation}\label{eq: optimal weights parabola}
	\w^*=\w_{GMV}+\frac{X-\RGMV }{s}\boldsymbol{Q}\bmu
	\end{equation}
where
\begin{equation}\label{eq: weights GMV}
\w_{GMV}=\frac{\bSigma\inv\1}{\1\T \bSigma\inv\1},
\end{equation}
are the weights of the global minimum variance portfolio; $\RGMV $ and $s$ are given in \eqref{eq: eff set const: R, V, s} and \eqref{eq:Q}; $X$ is the expected return of the optimal portfolio that maximizes the expected power utility as provided in \eqref{eq:X}.
\end{Cor}

The expression of optimal portfolio weights \eqref{eq: optimal weights parabola} coincides with the formula for weights obtained as a solution of Markowitz's portfolio selection problem. Moreover, the results of Corollary \ref{cor: parabola} are used to derive conditions on the relative risk aversion coefficient $\gamma$ which ensure that the obtained portfolio \eqref{eq: optimal weights parabola} lies in the upper part of the parabola \eqref{eq: Mean-Variance parabola} and, hence, the optimal portfolio that maximizes the expected power utility function is mean-variance efficient. This finding is summarized in Theorem \ref{th: power_utility_mv-efficiency} with the proof moved to the Appendix.

\begin{Th}\label{th: power_utility_mv-efficiency}
Under the conditions of Theorem \ref{th: optimal portfolio weights theorem}, the optimal portfolio in the sense of maximizing the expected power utility function is mean-variance efficient if and only if
\begin{eqnarray}\label{con: power_utility_mv-efficiency}
 && \gamma\geq\gamma_{min}  \quad \text{and} \quad R_{GMV}>0\,.
\end{eqnarray}
\end{Th}

The result of Theorem \ref{th: power_utility_mv-efficiency} provides the rigorous mathematical proof of Markowitz's conjecture that the mean-variance analysis provides a very good proxy to the utility optimization problem with the power utility in the sense that it sophistically approximates its solution (see, e.g., \citet{levy1979approximating, markowitz2014mean}). This result is also in-line with the finding of \cite{grauer_1986}. In the case of the power utility function, Corollary \ref{cor: parabola} shows that both approaches lead to the same set of optimal portfolios, while Theorem \ref{th: power_utility_mv-efficiency} makes one step further and presents the conditions under which the mean-variance efficiency holds for the optimal portfolios obtained by maximizing the expected power utility.

It is remarkable that the optimal portfolio in the sense of maximizing the expected power utility will never coincide with the global minimum variance portfolio (GMV). This observation follows from the fact that if $X=\RGMV $, then from Theorem \ref{th: optimal portfolio weights theorem} we get
\[Y=\frac{\gamma}{s}(X \RGMV  -\RGMV ^2-s\VGMV )={-\gamma \VGMV }<0,\]
what is not possible since $Y=E\left[(\w\T  \bR)^2\right] \ge 0$.

Further interesting property of the optimal portfolio in the sense of maximizing the expected power utility is formulated in Corollary \ref{cor: Sharpe ratio if gamma=infinity}. Here, we prove that it converges to the optimal portfolio that maximizes the Sharpe ratio, another popular portfolio in the financial literature, when $\gamma \to \infty$.

\begin{Cor}\label{cor: Sharpe ratio if gamma=infinity}
  Suppose that the assumptions of Theorem \ref{th: power_utility_mv-efficiency} are satisfied.
  If $\gamma$ tends to infinity, then the optimal portfolio that maximizes the expected power utility converges to the Sharpe ratio portfolio with the weights:
	\begin{equation}\label{eq: optimal weights at infinity}
	\w^*=\frac{\bSigma\inv\bmu}{\1\T \bSigma\inv\bmu}.
	\end{equation}
      \end{Cor}
        \begin{remark}\label{sharpe}
        \rm   Note that the Sharpe ratio portfolio presented in \eqref{eq: optimal weights at infinity} is slightly different from a classical one because the mean vector $\bmu$ is not equal to expected value of the asset returns $\text{E}(\mathbf{r})$ but to $\text{E}(\mathbf{R})=\text{E}(\mathbf{r})+\mathbf{1}$. This implies a very interesting fact that although the expected power utility portfolio can not be exactly equal to the GMV portfolio they are very close to each other in some situations. Indeed, if $\text{E}(\mathbf{r})$ is close to zero then the weights presented in \eqref{eq: optimal weights at infinity} are close to the GMV portfolio from \eqref{eq: weights GMV}. This is natural because the GMV portfolio plays the role of the least risky portfolio in the mean-variance framework.
        \end{remark}
Next, Corollary \ref{cor: Sharp greater} proves that the expected return of the optimal portfolio in the sense of maximizing the expected power utility function is not smaller than the expected return of the Sharpe ratio portfolio.

\begin{Cor}\label{cor: Sharp greater}
  Suppose that the assumptions of Theorem \ref{th: power_utility_mv-efficiency} are satisfied and $\RGMV  > 0$.
  Then the expected return of the optimal portfolio that maximizes the expected power utility is larger than or equal to the expected return of the Sharpe ratio portfolio, i.e.
\begin{equation}\label{ineq: Sharpe ratio cor}
	X \ge \frac{\bmu\T \bSigma\inv\bmu}{\1\T \bSigma\inv\bmu}\,.
\end{equation}
\end{Cor}
	The other interesting statement is related to the coefficient of relative risk aversion. It is natural to conclude that if $\gamma$  increases (what means the less risky investor) we can expect the decrease of the variance of the portfolio return. So as the result we present described fact in a Corollary \ref{cor: Decrease of variance with the increase of gamma}. Besides, we have also received that the expected portfolio return decreases as well.
	   
\begin{Cor}\label{cor: Decrease of variance with the increase of gamma}
  Suppose that the assumptions of Theorem \ref{th: power_utility_mv-efficiency} are satisfied and $\RGMV  > 0$.
  Then the expected return and the variance of the optimal portfolio that maximizes the expected power utility are decreasing functions of $\gamma$.
\end{Cor}

\subsection{Analytical solution for the logarithmic utility}

In Theorem \ref{th: optimal portfolio gamma=1} we present the closed-form solution of the optimization problem \eqref{eq:investor's problem} for the logarithmic utility given in \eqref{eq: log utility}. The proof of the theorem is given in the Appendix.

\begin{Th} \label{th: optimal portfolio gamma=1}
Assume that $\mathbf{(A1)}$ holds. If
\begin{equation}\label{ineq: condition_cor}
  \gamma_{min}\leq 1\,,
\end{equation}
where $\gamma_{min}$ is defined in \eqref{gamma_min}, 
then the solution of the optimization problem \eqref{eq:investor's problem} with the logarithmic utility function $U(\cdot)$ as in \eqref{eq: log utility} is given by 	
	\begin{equation}\label{eq:optimal portfolio weights_cor}
	\w^*=\bSigma \inv\left[\left(-\1+2\frac{{\bmu} } {X}\right)Y-X{\bmu} \right]
	\end{equation}
	with
\begin{eqnarray}\label{eq:XY_cor}
  X&=&\frac{3\RGMV  - \sqrt{\D}}{2(1+s)} \quad \text{and} \quad Y=\frac{X \RGMV  -\RGMV ^2}{s} -\VGMV,\\
  \D&=&9\RGMV ^2-8(1+s)(\RGMV ^2+s\VGMV )\nonumber\,.
	\end{eqnarray}
\end{Th}

The expected return of the optimal portfolio in the sence of maximizing the expected logarithmic utility is $X$ and its variance is given by $V=Y-X^2$. Using the results of Theorem \ref{th: optimal portfolio weights theorem}, we also obtain the maximum value of the optimization problem \eqref{eq:investor's problem} for the logarithmic utility given by	
	\begin{equation}\label{eq: max expected utility gamma=1}
	\max_{\w^*}E\left[U(W )\right]=\ln W_0+2\ln X -\frac{1}{2}\ln Y.
	\end{equation}

The expression of the weights obtained for the logarithmic utility function is a special case of the optimal portfolio weights derived for the power utility function corresponding to $\gamma=1$. This finding is not surprising since the logarithmic utility function is a limiting case of the power utility function when $\gamma \to 1$. Although the power utility cannot be defined for $\gamma=1$ and only its limiting behavior is considered, that is not longer the case with the formula for the weights presented in Theorem \ref{th: optimal portfolio weights theorem} which can also be computed for $\gamma=1$ as soon as the condition \eqref{gamma_min} is fulfilled which coincides with \eqref{ineq: condition_cor}  of Theorem \ref{th: optimal portfolio gamma=1}.

Another important application of the results given in Theorems \ref{th: optimal portfolio weights theorem} and \ref{th: optimal portfolio gamma=1} is that the weights of the optimal portfolio in the sense of maximizing the expected logarithmic utility function can also be presented in the form of Markowitz's portfolio, similarly to the expression for the weights obtained for the power utility function and they are given by
\begin{equation}\label{eq: optimal weights parabola log}
	\w^*=\w_{GMV}+\frac{X-\RGMV }{s}\boldsymbol{Q}\bmu
\end{equation}
with $X$ as in \eqref{eq:XY_cor}, i.e. the optimal portfolio that maximizes the expected logarithmic utility function also lies on parabola \eqref{eq: Mean-Variance parabola}. Finally, Theorem \ref{th: log_utility_mv-efficiency} presents conditions under which this portfolio is mean-variance efficient, i.e. it lies on the upper part of the parabola.

\begin{Th}\label{th: log_utility_mv-efficiency}
  Under the conditions of Theorem \ref{th: optimal portfolio weights theorem}, the optimal portfolio in the sense of maximizing the expected logarithmic utility is mean-variance efficient if and only if
\begin{equation}\label{con: log_utility_mv-efficiency}
\gamma_{min}\leq 1 \quad \text{and} \quad R_{GMV}>0\,.
\end{equation}
\end{Th}

\section{Empirical Study}
The aim of this section is to show how the above results can be applied in practice. We will analyze in detail as well whether the demanded assumptions are fulfilled in practice or not.

Our empirical study is based on stocks listed in the German stock index (DAX). We  {select} 17 stocks from the DAX {(ADS, ALV, BAYN, BMW, CBK, DAI, DBK, DPW, DTE, HEI, IFX, LHA, LIN,   SAP, SIE, TKA, VOW3)} and study  {their} behavior   {during} the time period from September 21, 2014 up to September 17, 2017. Our analysis is based on weekly returns. Consequently we have 157 observations for each stock.

Based on these 17 stocks we construct a number of portfolios by randomly sampling $k \in \{4,...,14\}$ assets. Thus, for each $k$ the number of different sets of assets is equal to $\binom{17}{k}$. From every such set we determine the optimal portfolios in the sense of maximizing the expected power (logarithmic) utility. The parameters $\bmu$ and $\bSigma$ are estimated for each portfolio constellation by the sample mean vector and the sample covariance matrix. Moreover, without loss of generality we set $W_0 = 1$.

\subsection{Validating of model assumption and conditions of theorems}

The findings of Theorems \ref{th: optimal portfolio weights theorem} and \ref{th: optimal portfolio gamma=1} are obtained under certain assumptions. In the following we analyze whether these conditions are reasonable for the underlying data set.

First, we validate the assumption of log-normally distributed portfolio returns. In order to test the hypothesis of log-normality, we apply the Shapiro-Wilk test to the logarithmized returns of optimal portfolios that maximize the expected power utility functions for the chosen set of assets. In Table \ref{table1} the 25\% quantiles of the $p$-values are provided for several numbers of assets $k$ within the portfolio and for several values of the relative risk aversion coefficient $\gamma$. We remind again that for each fixed $k$ the amount of optimal portfolios under consideration is $\binom {17}{k}$. Table \ref{table1} shows that more than 75\% of the $p$-values are always considerably larger than the possible nominal levels of the test, like  5\% in our case. Consequently, in at least of 75\% of all considered cases, the null hypothesis that the returns of optimal portfolios in the sense of maximizing the expected power utility are log-normally distributed cannot be rejected. For example, choosing $k=12$ and $\gamma = 5$ we receive that 75\% of $p$-values are greater than $0.26$. Moreover, with increasing number of assets the $p$-values are monotonically increasing as well.
Hence, this gives rise to the conclusion that the log-normality assumption seems to be reasonable in a number of practical situations.

\begin{table}
	\centering
	\begin{tabular}{|c|l|l|l|l|l|l|l|l|l|}
		\hline
		\multirow{2}{*}{\begin{tabular}[c]{@{}c@{}}Number of \\ assets $i$ \end{tabular}} & \multicolumn{9}{c|}{Relative risk aversion $\gamma$}                                                                                                                                                                 \\ \cline{2-10}
		& \multicolumn{1}{c|}{2} & \multicolumn{1}{c|}{3} & \multicolumn{1}{c|}{4} & \multicolumn{1}{c|}{5} & \multicolumn{1}{c|}{6} & \multicolumn{1}{c|}{7} & \multicolumn{1}{c|}{8} & \multicolumn{1}{c|}{9} & \multicolumn{1}{c|}{10} \\ \hline
		4  & 0.07 & 0.09 & 0.10 & 0.10 & 0.11 & 0.12 & 0.13 & 0.13 & 0.14 \\ \hline
		5  & 0.10 & 0.11 & 0.12 & 0.12 & 0.13 & 0.14 & 0.15 & 0.15 & 0.16 \\ \hline
		6  & 0.12 & 0.14 & 0.13 & 0.14 & 0.14 & 0.15 & 0.16 & 0.17 & 0.17 \\ \hline
		7  & 0.13 & 0.17 & 0.16 & 0.15 & 0.15 & 0.16 & 0.17 & 0.17 & 0.18 \\ \hline
		8  & 0.14 & 0.20 & 0.18 & 0.16 & 0.16 & 0.17 & 0.17 & 0.17 & 0.17 \\ \hline
		9  & 0.14 & 0.23 & 0.21 & 0.18 & 0.17 & 0.17 & 0.17 & 0.16 & 0.16 \\ \hline
		10 & 0.15 & 0.25 & 0.24 & 0.21 & 0.19 & 0.17 & 0.16 & 0.16 & 0.16 \\ \hline
		11 & 0.15 & 0.28 & 0.27 & 0.23 & 0.20 & 0.18 & 0.16 & 0.15 & 0.15 \\ \hline
		12 & 0.16 & 0.30 & 0.29 & 0.26 & 0.22 & 0.20 & 0.17 & 0.15 & 0.15 \\ \hline
		13 & 0.19 & 0.32 & 0.31 & 0.29 & 0.25 & 0.21 & 0.18 & 0.16 & 0.15 \\ \hline
		14 & 0.23 & 0.34 & 0.34 & 0.33 & 0.31 & 0.25 & 0.21 & 0.17 & 0.16 \\ \hline
	\end{tabular}
	\caption{25$\%$ quantile of the $p$-values of the log-normality test for several combinations of $k$ and $\gamma$}
	\label{table1}
\end{table}

Next, we analyze whether the condition \eqref{gamma_min} of Theorem \ref{th: optimal portfolio weights theorem} is fulfilled in the considered application. Our results are summarized in Figure \ref{fig: conditions of optimal portfolio theorem } (above). Here the percentage of all considered cases are present for several values of the number of assets $k$ in the portfolio and of the relative risk aversion $\gamma$ when the condition $\{ \gamma\geq \gamma_{min}\}$ is not fulfilled. The dashed vertical lines represent the values of $\gamma_{min}(k)$ for every number of assets $k$. We observe that with decreasing $k$ and increasing $\gamma$ this condition is almost always satisfied and the number of inappropriate cases goes to zero. For example, the probability of having $\gamma<\gamma_{min}$ for $k$ assets is virtually zero if $\gamma>\gamma_{min}(k)$ for all considered values of $k\in\{4, 5, \ldots, 14\}$. This is in line with the theoretical findings of Theorem \ref{th: optimal portfolio weights theorem}. Moreover, because $\gamma_{min}(k)$ is always smaller than one, we can conclude due to Theorem \ref{th: optimal portfolio gamma=1} that the solution for the logarithmic portfolios exists as well. Hence, the condition \eqref{gamma_min} is realistic for both small or large portfolios and for the relatively large values of the risk aversion coefficient, at the same time the optimal portfolio that maximizes the expected power utility exists with a small probability if $\gamma$ is close to zero. 

Finally, we analyze the condition for the mean-variance efficiency of the optimal portfolio that maximizes the expected power utility which is stated in Theorem \ref{th: power_utility_mv-efficiency}. The results are presented in Figure \ref{fig: conditions of optimal portfolio theorem } (below). Here we proceed in the same way as above for Theorem \ref{th: optimal portfolio weights theorem}. The second expression in \eqref{con: power_utility_mv-efficiency} supplies the classical  mean-variance efficiency in sense that the 
accumulated expected portfolio return of GMV portfolio, i.e., $R_{GMV}=\frac{\1\T \bSigma\inv\bmu}{\1\T \bSigma\inv\1}=\frac{\1\T \bSigma\inv\text{E}(\mathbf{r})}{\1\T \bSigma\inv\1}+1$  must be always positive. As we can observe, the related results look the same as the upper part of Figure \ref{fig: conditions of optimal portfolio theorem }, which means that the first expression of \eqref{con: power_utility_mv-efficiency} seems to be a stronger restriction than the second one for the considered data set.

\begin{figure}
  \includegraphics[width=\textwidth]{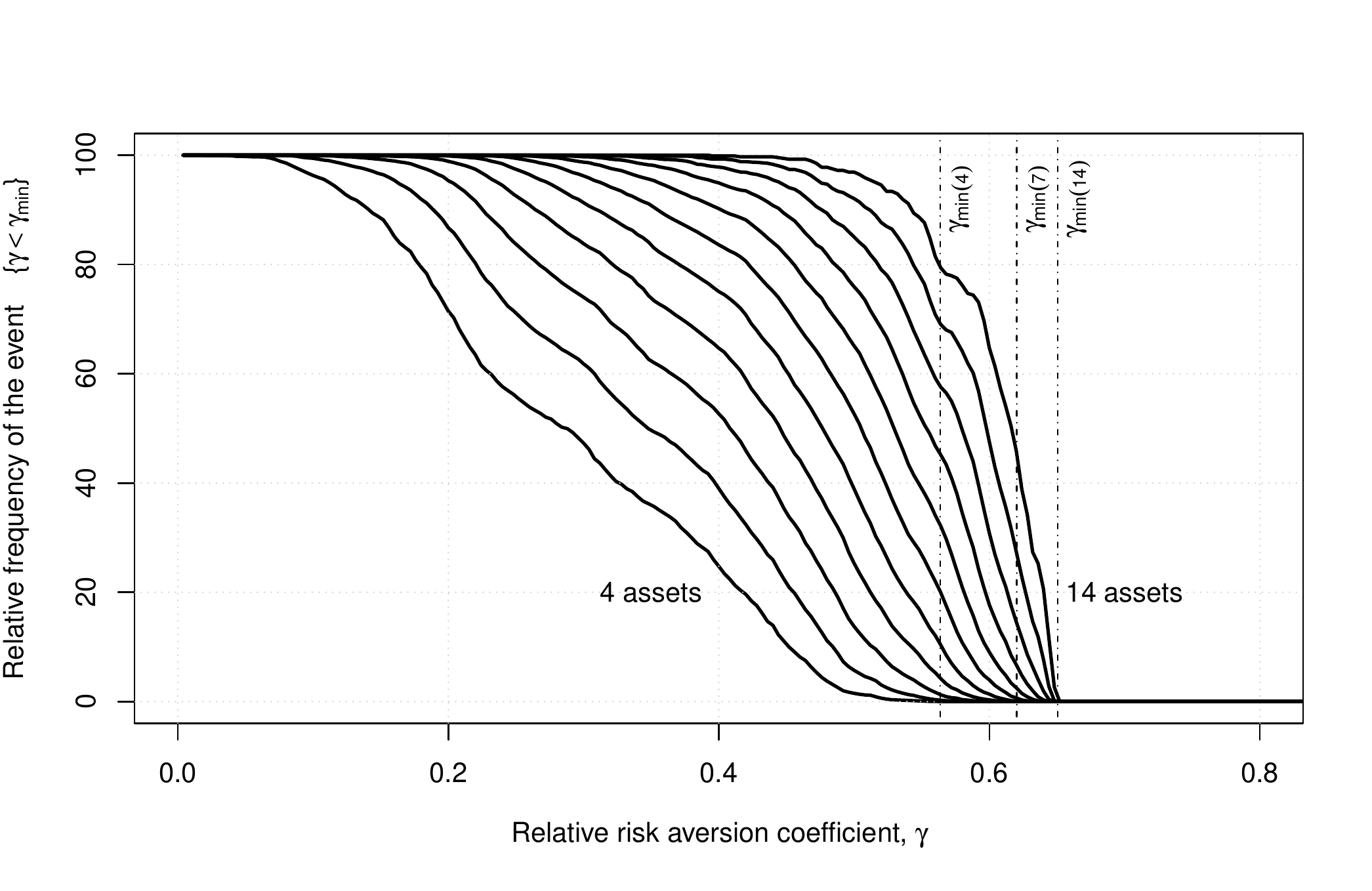} \includegraphics[width=\textwidth]{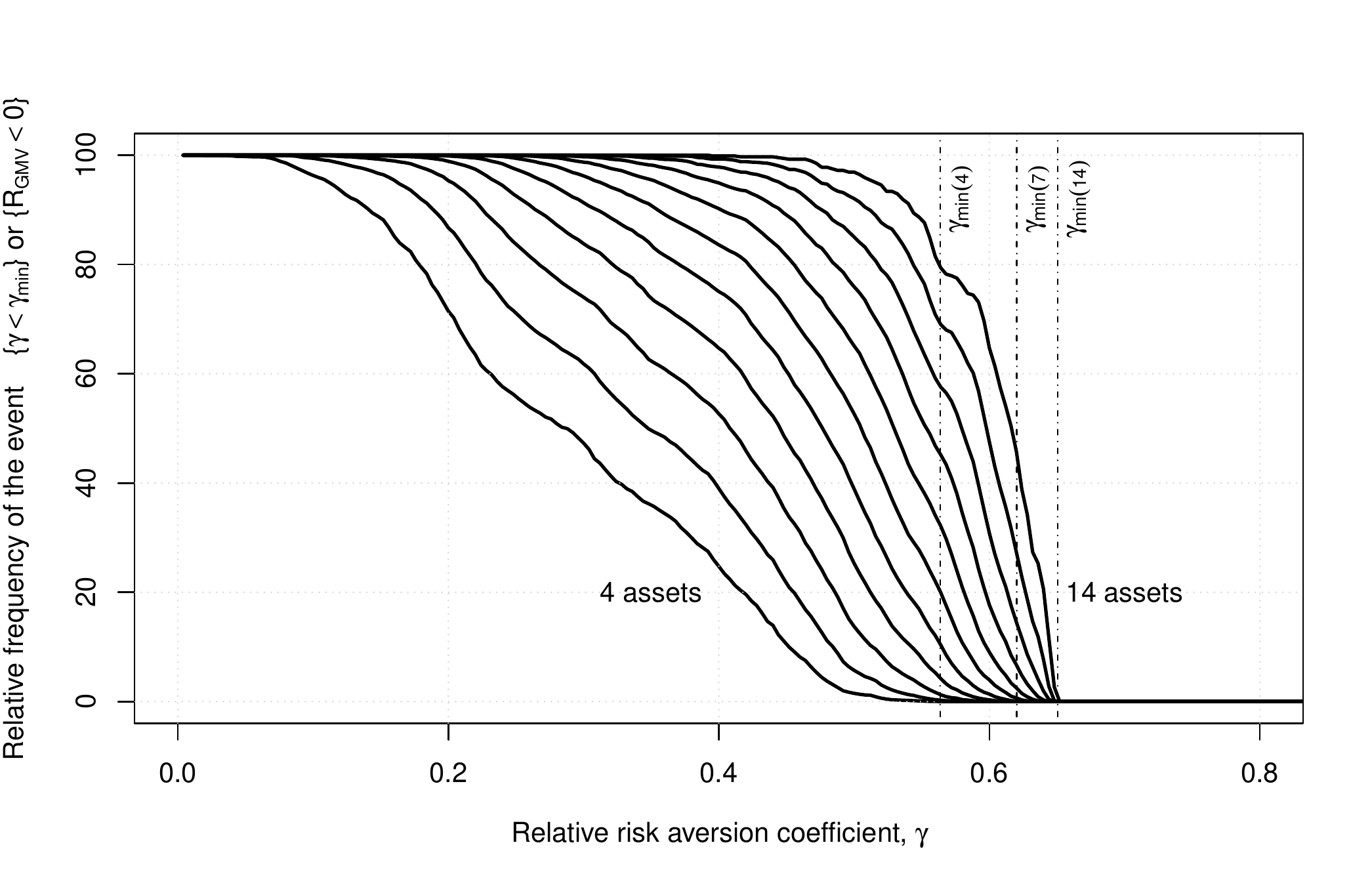}
  \caption{Relative number of cases when the condition $\gamma\geq \gamma_{min}$ of Theorem \ref{th: optimal portfolio weights theorem} (above) and the condition (\ref{con: power_utility_mv-efficiency}) of Theorem \ref{th: power_utility_mv-efficiency} (below) are not fulfilled plotted as a function of $\gamma$ for several number of assets $k \in \{4,5,...,14\}$.}
	\label{fig: conditions of optimal portfolio theorem }
\end{figure}


\subsection{Location of the optimal portfolio that maximizes the expected power utility function on the efficient frontier}

An other part of the empirical study presented by Figure \ref{fig: efficient frontiers k=4.7.10.14} is related to the location of the optimal portfolios. As it was shown in \eqref{eq:X} and \eqref{eq:Y} the mean and the variance of the optimal portfolio depends on three parameters: $\bmu, \bSigma, \gamma$. The mean vector and the covariance matrix  can be estimated and in our example we consider four possible number of assets $k\in\{4, 7, 10, 14\}$. We choose the first $k$ assets in the alphabetic order. Then we build efficient frontier for each combination of assets where we add the location of the optimal portfolios for $\gamma\in\{1, 2, 5\}$ and the Sharpe ratio portfolio. Like it was proved in Corollary \ref{cor: Decrease of variance with the increase of gamma} we receive the decrease of the variance and the mean with the increase of a risk aversion coefficient and optimal portfolio tends to the Sharpe ratio portfolio as it was shown in Corollary \ref{cor: Sharpe ratio if gamma=infinity}. In addition, we also receive the increase of both mean and variance of portfolio returns with a larger amount of assets in the portfolio.

\begin{figure}
	\includegraphics[width=0.5\textwidth]{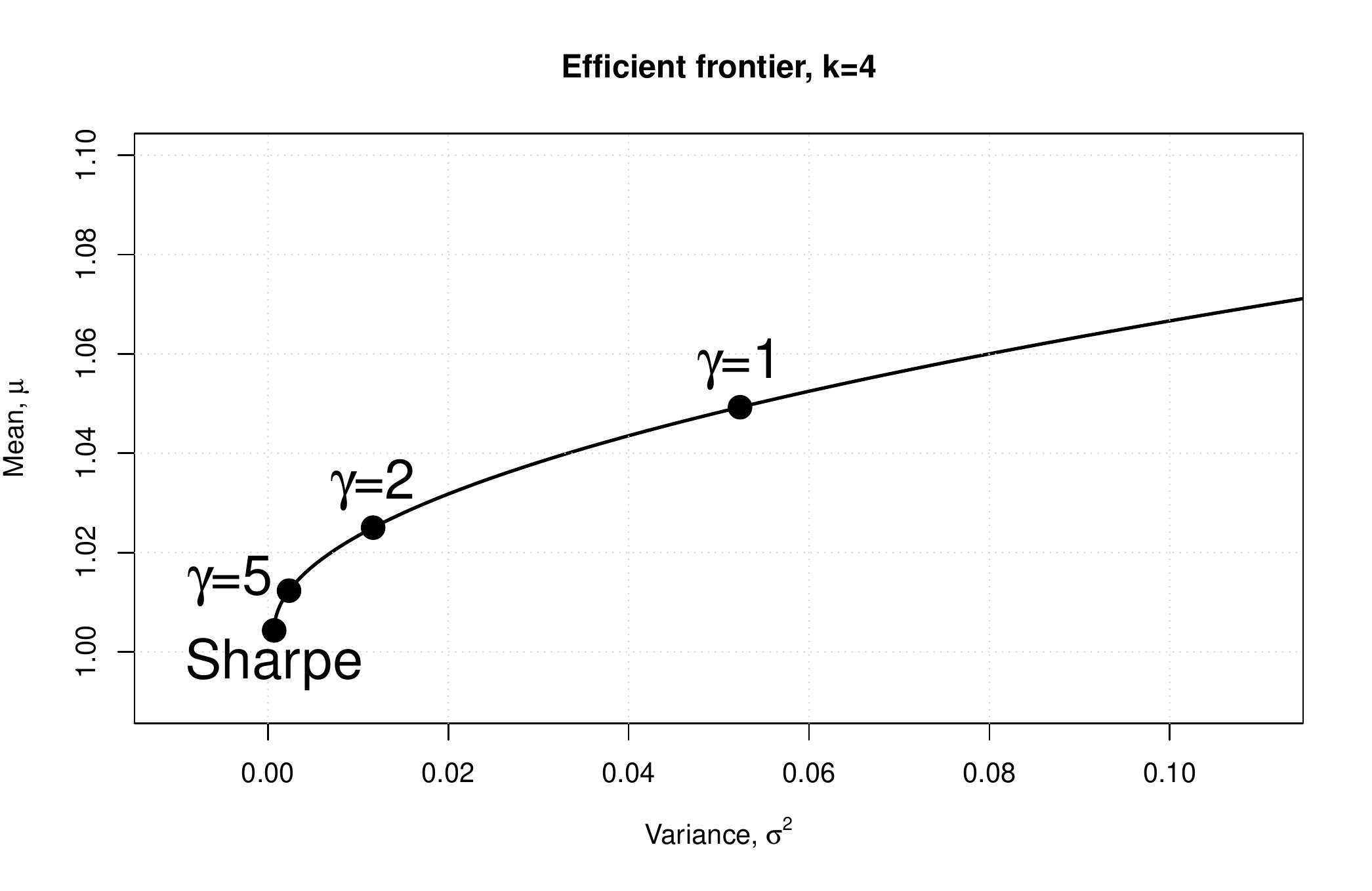}
	\includegraphics[width=0.5\textwidth]{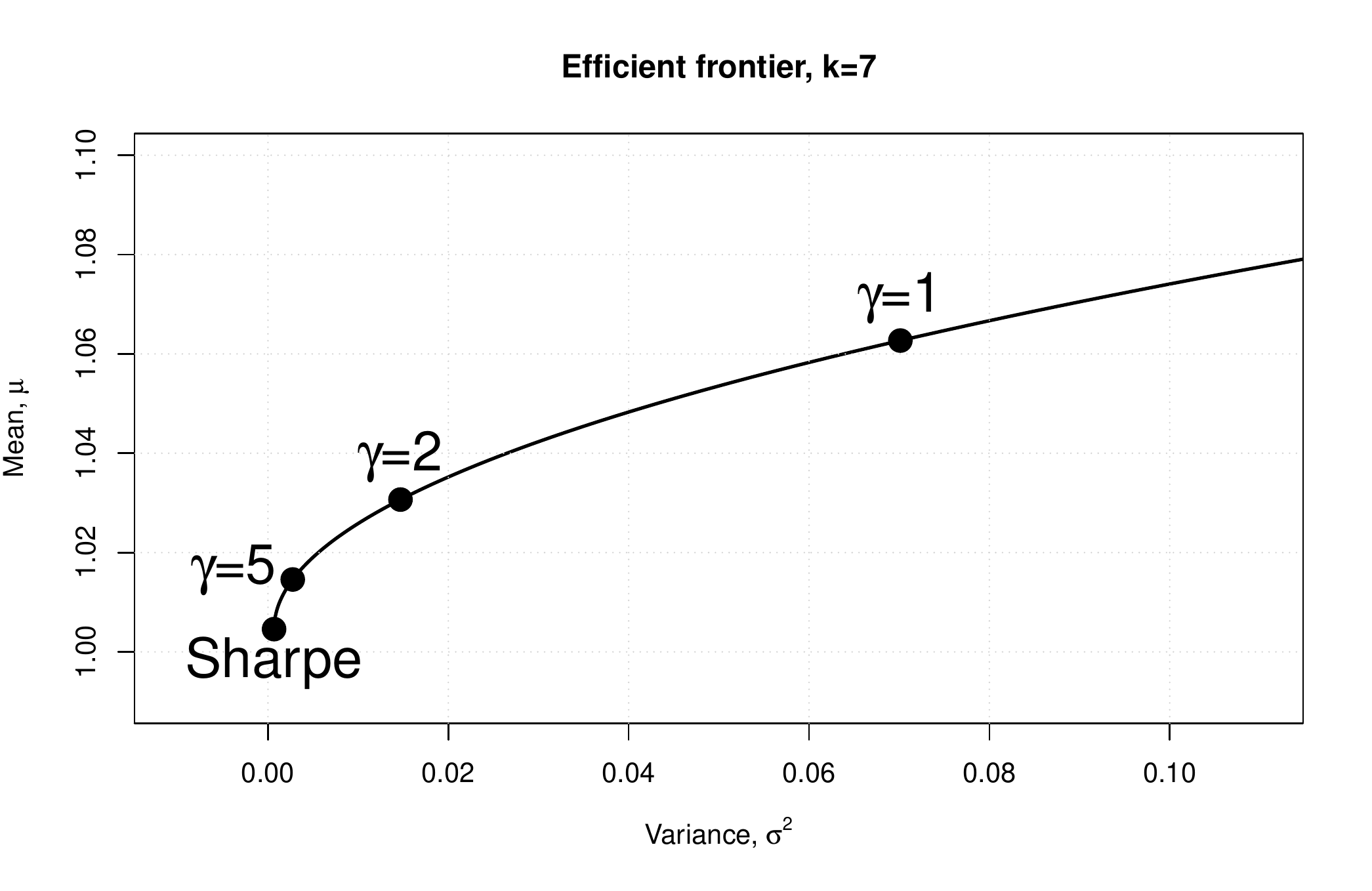}
	\includegraphics[width=0.5\textwidth]{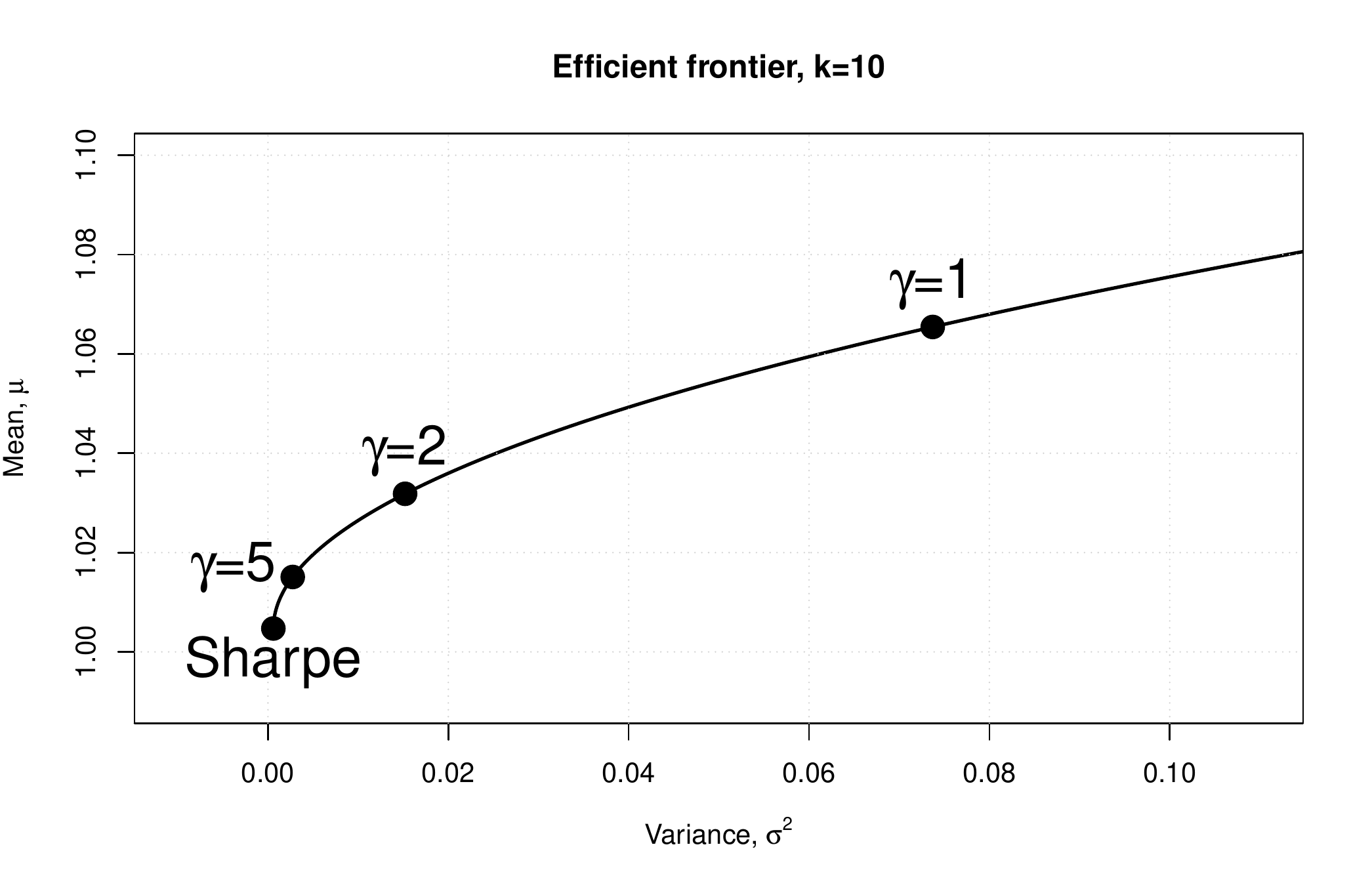}
	\includegraphics[width=0.5\textwidth]{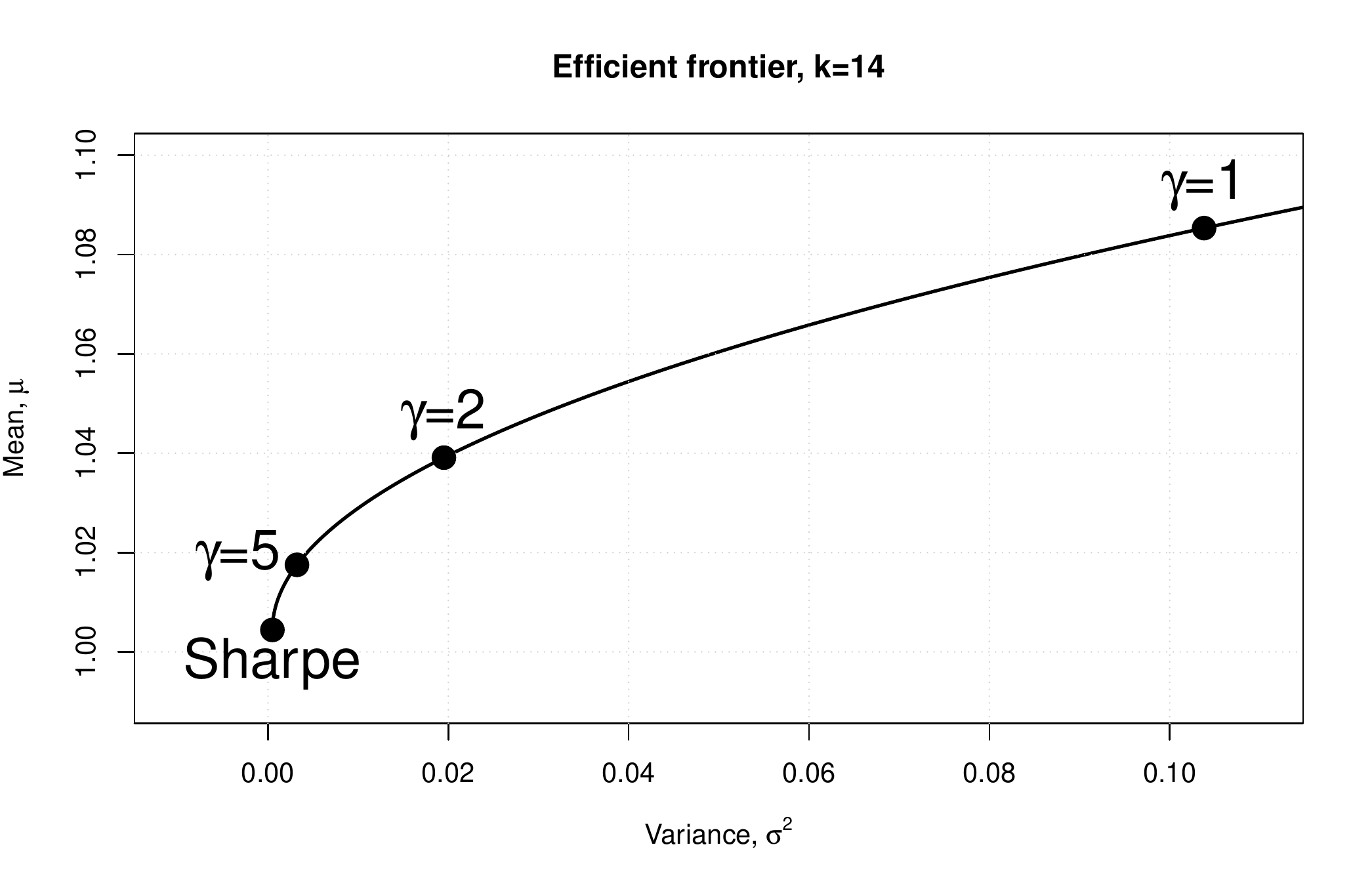}
	\caption{{Optimal portfolios for $\gamma\in\{1, 2, 5\}$ and the Sharpe ratio portfolio located on the efficient frontier for a number of assets $k\in\{4, 7, 10, 14\}$.}}
	\label{fig: efficient frontiers k=4.7.10.14}	
\end{figure}

\subsection{A comparison of several portfolio strategies}

Finally, we compare the power utilities of several portfolio strategies with each other. We examine the naive portfolio with equal weights, the Sharpe ratio portfolio, and the optimal portfolio derived in Theorem \ref{th: optimal portfolio weights theorem}. Note that the GMV portfolio was excluded from the study because it was to close to Sharpe ratio portfolio for a considered data set (see, Remark \ref{sharpe}). In Figure \ref{fig: portfolio comparison} the empirical distribution functions of the power utility obtained for each of the considered methods are given for the portfolios consisting of $k=9$ assets and for various values of $\gamma$. The figure shows that the best results are obtained for the strategy derived in Theorem \ref{th: optimal portfolio weights theorem} as expected. In addition, as another visualisation of Corollary \ref{cor: Sharpe ratio if gamma=infinity} result we have that the utility of the Sharpe ratio portfolio nearly coincides with the derived one in Theorem \ref{th: optimal portfolio weights theorem} as the coefficient of a relative risk aversion increases. Finally, a very bad performance of the equally-weighted (naive) portfolio is observed.

\begin{figure}
	\includegraphics[width=0.5\textwidth]{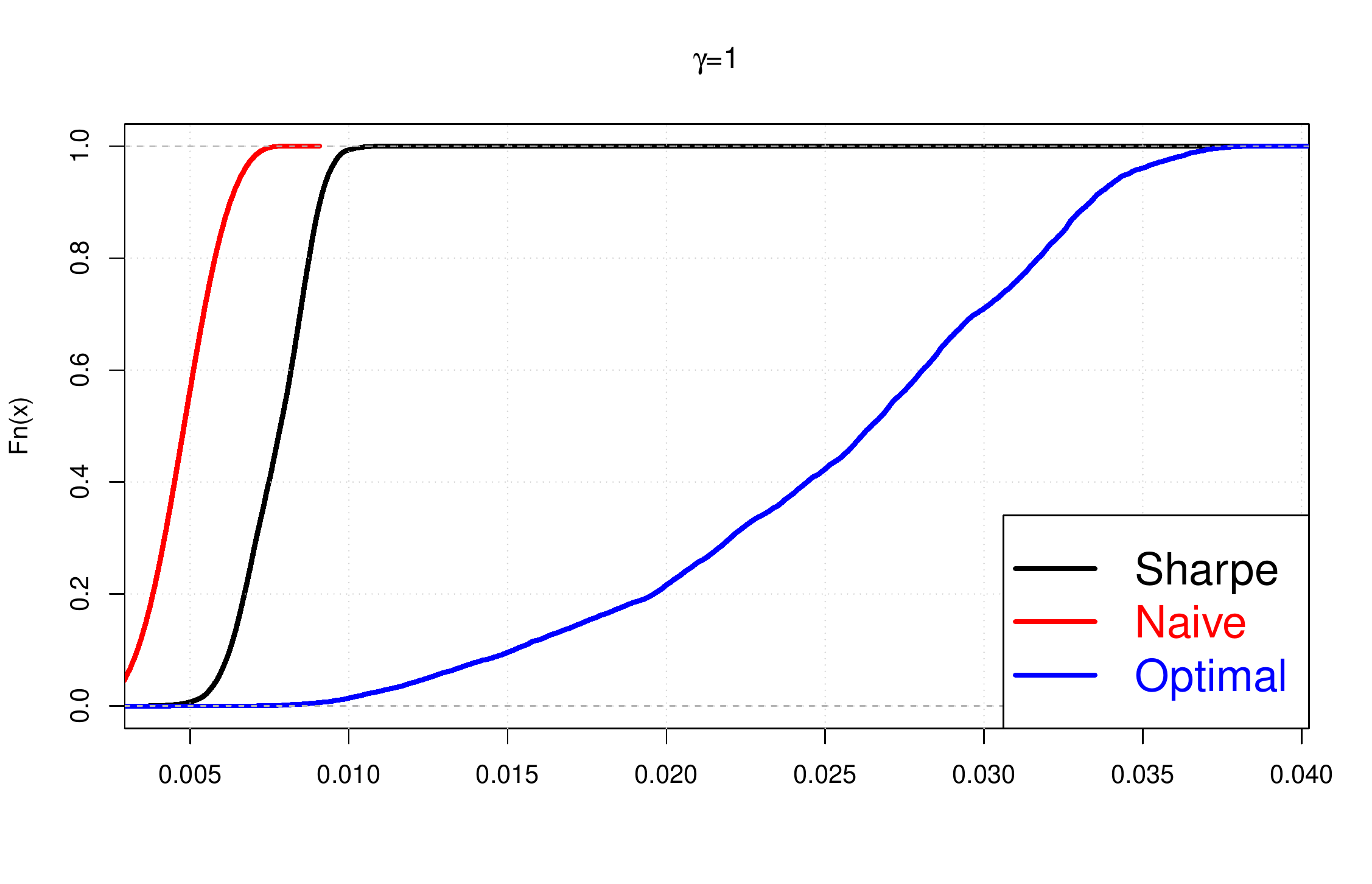}
	\includegraphics[width=0.5\textwidth]{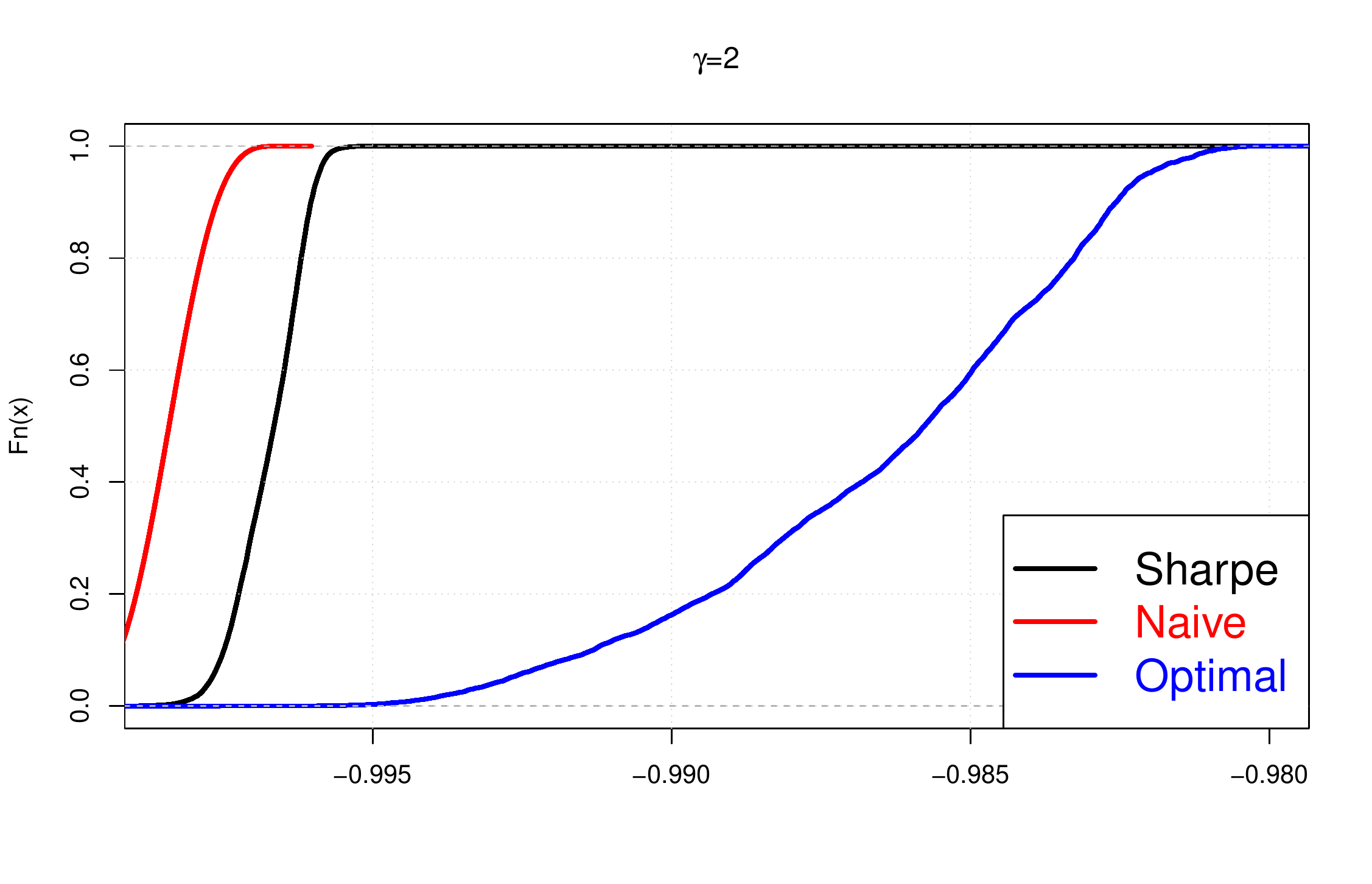}
	\includegraphics[width=0.5\textwidth]{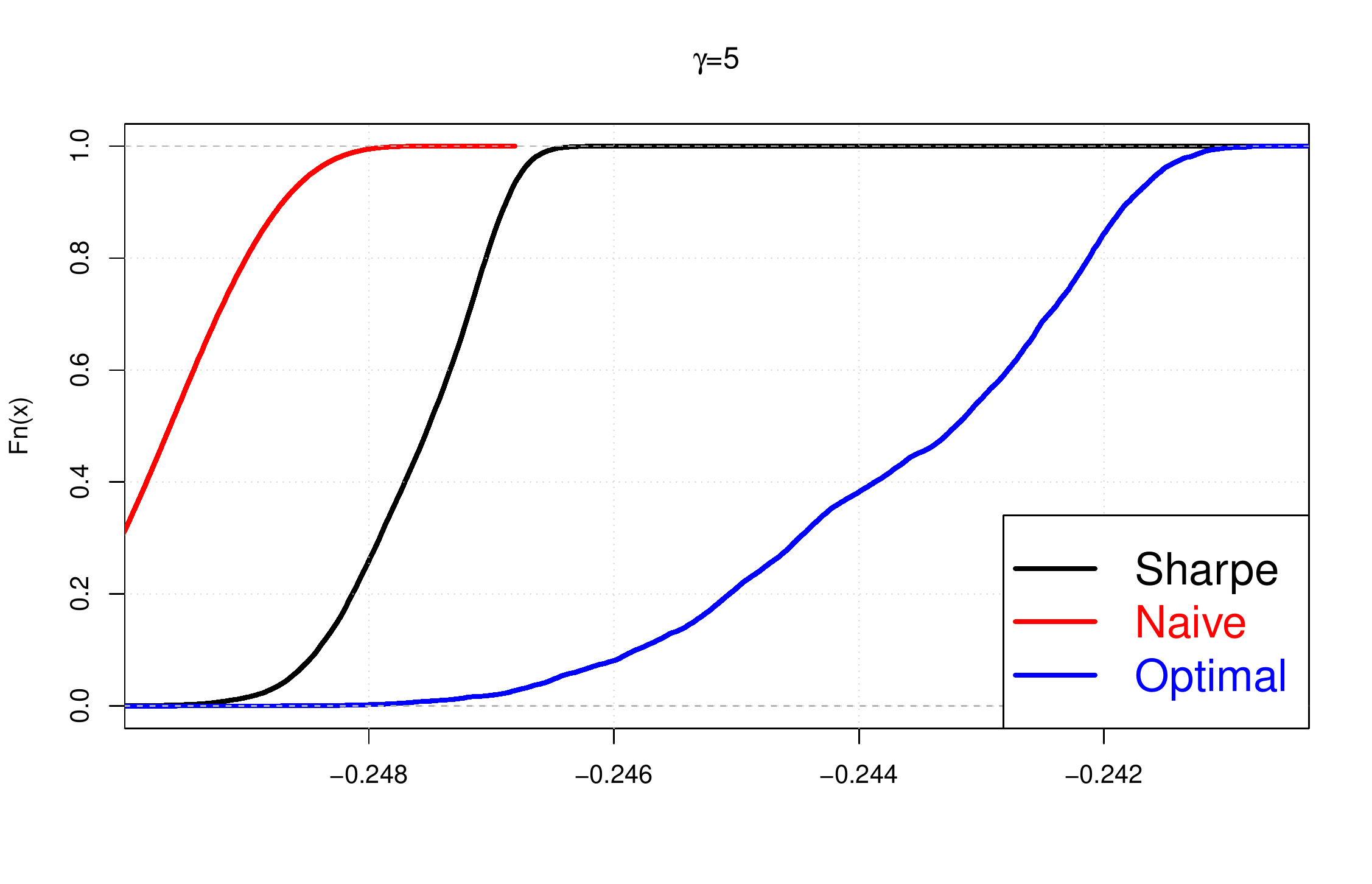}
	\includegraphics[width=0.5\textwidth]{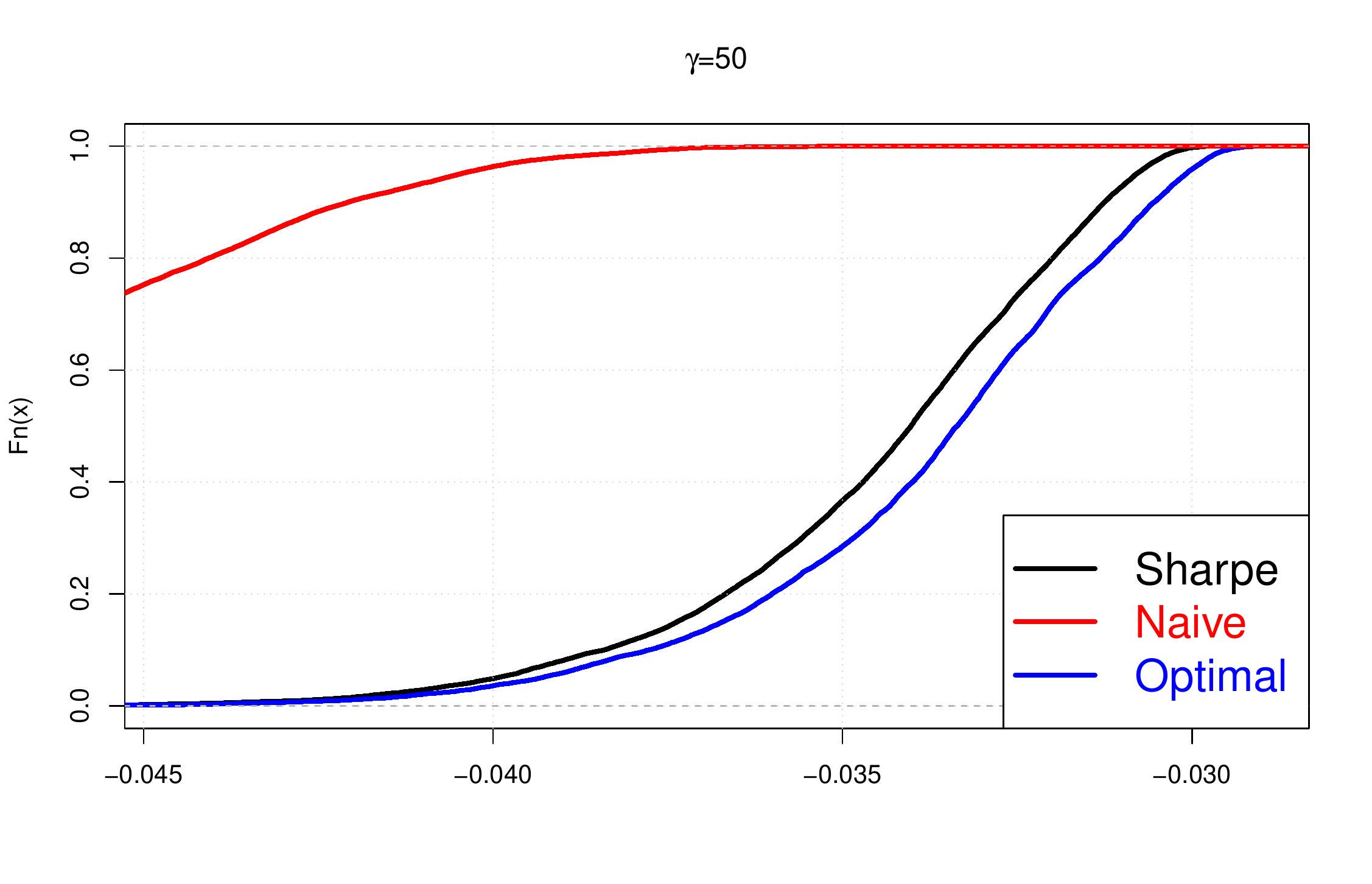}
	\caption{Empirical distribution functions of the power utility for several portfolio strategies (naive portfolio, Sharpe portfolio, optimal portfolio in the sense of maximizing the expected power utility) for $k=9$ assets.}
	\label{fig: portfolio comparison}	
\end{figure}

\section{Appendix}

\begin{proof}[Proof of Lemma \ref{lem0}]
It is easy to see that $\Psi(\mu)=0$ and the point $\mu$ is its minimum. Indeed, let us consider a point different from $\mu$ and define it as $x_0=a\mu$, $a>0$, $a\neq1$. Then then it holds
\begin{equation*}
\begin{split}
\Psi(x_0)&=\Phi\left(\frac{a\mu-\mu}{\sigma}\right)-\Phi\left(\frac{\ln a\mu-\ln\mu}{\sigma\slash\mu}\right)=\Phi\left(\frac{a-1}{\sigma\slash\mu}\right)-\Phi\left(\frac{\ln a}{\sigma\slash\mu}\right)\\
&>\Phi\left(\frac{a-1}{\sigma\slash\mu}\right)-\Phi\left(\frac{a-1}{\sigma\slash\mu}\right)=0.
\end{split}
\end{equation*}

Now we consider the upper bound of $\Psi(x)$

\begin{equation}
\frac{\partial}{\partial x}\Psi(x)=\frac{1}{\sqrt{2\pi}\sigma}e^{-\frac{(x-\mu)^2}{2\sigma^2}}-
\frac{1}{x\sqrt{2\pi}\sigma\slash\mu}e^{-\frac{(\ln x-\ln\mu)^2}{2\sigma^2\slash\mu^2}}=0
\end{equation} 
or
\begin{equation}
\left(\frac{x}{\mu}-1\right)^2=\ln^2\frac{x}{\mu}+2\frac{\sigma^2}{\mu^2}\ln\frac{x}{\mu}.
\end{equation}  
Denoting $y=\ln\frac{x}{\mu}$ we find
\begin{equation}\label{eq: extrema points of Psi}
\left(e^y-1\right)^2=y^2+2\frac{\sigma^2}{\mu^2}y.
\end{equation}  
One can see that the equation \eqref{eq: extrema points of Psi} has three roots (plotting is helpful in this case): the one is grater that zero, the second one is less than zero and already examined $y=0$. It is also notable that as soon as $\Psi(\cdot)$ is a continuously differentiable function and $\Psi(x)\rightarrow0\ \text{for}\ x\rightarrow\pm\infty$, non zero extrema points are local and global maxima.    

For $y<0$ the expression $(e^y-1)^2$ is smaller than $\left(\frac{y}{y-1}\right)^2$ and if $y<-2\frac{\sigma^2}{\mu^2}$ holds that $y^2+2\frac{\sigma^2}{\mu^2}y>\left(y+2\frac{\sigma^2}{\mu^2}\right)^2$, consequently, the solution belongs to the interval $\left(1-\frac{\sigma^2}{\mu^2}-\sqrt{\frac{\sigma^4}{\mu^4}+1},-2\frac{\sigma^2}{\mu^2}\right)$ which can be derived plugging left hand side of \eqref{eq: extrema points of Psi} equal to zero and solving the equation $\left(\frac{y}{y-1}\right)^2=\left(y+2\frac{\sigma^2}{\mu^2}\right)^2$ for a negative root. 

In case of $y>0$ it holds that $(e^y-1)^2>\left(\left(1+\frac{y}{2}\right)^2-1\right)^2$ and $y^2+2\frac{\sigma^2}{\mu^2}y<\left(y+\frac{\sigma^2}{\mu^2}\right)^2$ and one can derive that the other extrema locates between zero and $2\frac{\sigma}{\mu}$.
As the result we have two intervals for extrema points\\ $l_1=\left(e^{1-\frac{\sigma^2}{\mu^2}-\sqrt{\frac{\sigma^4}{\mu^4}+1}},e^{-2\frac{\sigma^2}{\mu^2}}\right)$ and $l_2=\left(1,e^{2\frac{\sigma}{\mu}}\right)$, so $\Psi(x)$ can be bounded as follows.
\begin{eqnarray*}\label{psi_estimate}
\sup_{x\in \mathbb{R}}|\Psi(x)|&\leq&\sup_{x\in \mathbb{R}}\int\limits^{\frac{x-\mu}{\sigma}}_{\frac{\ln x-\ln\mu}{\sigma\slash\mu}}\frac{1}{\sqrt{2\pi}}|e^{-\frac{t^2}{2}}|dt\leq\sup_{x\in \mathbb{R}}\int\limits^{\frac{x-\mu}{\sigma}}_{\frac{\ln x-\ln\mu}{\sigma\slash\mu}}\frac{1}{\sqrt{2\pi}}dt\nonumber\\
&=&\frac{1}{\sqrt{2\pi}}\sup_{x\in \mathbb{R}}\frac{\frac{x}{\mu}-1-\ln\frac{x}{\mu}}{\sigma\slash\mu}=\frac{1}{\sqrt{2\pi}}\sup_{\frac{x}{\mu}\in l_1\cup l_2}\frac{\frac{x}{\mu}-1-\ln\frac{x}{\mu}}{\sigma\slash\mu}\nonumber\\
                               &=&\frac{1}{\sqrt{2\pi}}\max\left[\frac{e^{1-\frac{\sigma^2}{\mu^2}-\sqrt{\frac{\sigma^4}{\mu^4}+1}}-2+\frac{\sigma^2}{\mu^2}+\sqrt{\frac{\sigma^4}{\mu^4}+1}}{\frac{\sigma}{\mu}};\ \frac{e^{2\frac{\sigma}{\mu}}-1-2\frac{\sigma}{\mu}}{\frac{\sigma}{\mu}}\right]\,.
\end{eqnarray*}
Now taking the limit $\sigma/\mu\to 0$ we get
\begin{eqnarray}
  \sup_{x\in \mathbb{R}}|\Psi(x)|=O\left(\frac{\sigma}{\mu} \right)\,,
\end{eqnarray}
where the last equality folows from the fact that $e^{x}=1+x+O(x^2)$ as $x\to0$.
\end{proof}  

\begin{lemma}\label{lem1}
Assume that the portfolio return $\w\T  \bR$ at the end of the investment period is log-normally distributed, i.e. $\w\T  \bR\sim \ln \, N(\alpha, \beta^2)$. Let
      \begin{eqnarray}\label{eq:E,V}
	E & := & E\left[\w\T  \bR\right]=\w\T \bmu , \\
	V & :=& Var\left[\w\T  \bR\right]=\w\T \bSigma\w.\nonumber
	\end{eqnarray}
Then
	\begin{eqnarray}\label{eq: alpha, betha}
\beta ^2  =  \ln\left(\frac{V}{E^2}+1\right) \quad \text{and} \quad
\alpha =  \ln\frac{E^2}{(V+E^2)^{\frac{1}{2}}}  . \nonumber
	\end{eqnarray}
\end{lemma}

\begin{proof}[Proof of Lemma \ref{lem1}]
In using \eqref{eq: moments of log-normal distribution}, we get
 \begin{eqnarray}\label{eq:E,V_proof}
	 E\left[\w\T  \bR\right]&=&e^{\alpha+\frac{1}{2}\beta^2} , \\
	Var\left[\w\T  \bR\right]&=&e^{2\alpha+\beta^2}(e^{\beta^2}-1) .\nonumber
	\end{eqnarray}
Substituting \eqref{eq:E,V} into \eqref{eq:E,V_proof} and solving these two equations with respect to $\alpha$ and $\beta^2$ leads to the statement of the lemma.
\end{proof}

\begin{proof}[Proof of Theorem \ref{th: optimal portfolio weights theorem}]
	For the power utility function, we get that 	
	\begin{eqnarray}\nonumber
	&&E\left[U(W )\right] =  \frac{{W_0^{1-\gamma}}}{1-\gamma} \exp \left[\alpha(1-\gamma) + \frac{1}{2}\beta^2(1-\gamma)^2\right]\\
	&=&\frac{{W_0^{1-\gamma}}}{1-\gamma}\exp\left[(1-\gamma)\ln\frac{E^2}{(V+E^2)^{\frac{1}{2}}}+\frac{1}{2}(1-\gamma)^2\ln\frac{V+E^2}{E^2}\right] \nonumber\\
	&=&\frac{{W_0^{1-\gamma}}}{1-\gamma}\exp\left[(1-\gamma^2)\ln E+\frac{1}{2}(\gamma^2-\gamma)\ln(V+E^2)\right]\nonumber\\
	&=&\frac{{W_0^{1-\gamma}}}{1-\gamma}\exp\left[(1-\gamma^2)\ln \w\T {\bmu} +\frac{1}{2}(\gamma^2-\gamma)\ln\left(\w\T \bSigma \w+(\w\T {\bmu} )^2\right)\right].\label{eq:normal_expectation}
	\end{eqnarray}
	In order to maximize the expected utility we need to find the maximum of the expression in the exponent of (\ref{eq:normal_expectation}) for $\gamma<1$ or its minimum for $\gamma>1$ under the side condition $\w\T  \boldsymbol{1} = 1$. {In both cases the method of Lagrange multipliers is used with the Lagrange function given by}
	\begin{equation}\label{eq:maximization_problem}
	(1-\gamma^2)\ln \w\T {\bmu} +\frac{1}{2}(\gamma^2-\gamma)\ln\left(\w\T \bSigma \w+(\w\T {\bmu} )^2\right) + \lambda (\w\T  \boldsymbol{1} - 1)	.
	\end{equation}
	Partial derivation leads to
	\begin{align}\label{eq: lagrange derivetive}
	&\frac{\partial}{\partial \w}L=(1-\gamma^2)\frac{{\bmu} } {\w\T {\bmu} }+(\gamma^2-\gamma)\frac{\bSigma \w+(\w\T {\bmu} ){\bmu} }{\w\T \bSigma \w+(\w\T {\bmu} )^2}+\lambda\1=\mathbf{0} ,
	\\
	&\frac{\partial}{\partial \lambda}L=\w\T \1-1=0 .	
	\end{align}
	
	{Let}
	\begin{equation}\label{XY}
	X:=\w\T {\bmu} ,\quad \text{and} \quad
	Y:=\w\T \bSigma \w+(\w\T {\bmu} )^2.
	\end{equation}
	Then, the multiplication of \eqref{eq: lagrange derivetive} by $\w\T $ leads to $\lambda = {\gamma-1}$. Furthermore, multiplying \eqref{eq: lagrange derivetive} by ${\bmu} \T \bSigma \inv/{\1\T \bSigma \inv\1}$ and $\1\T \bSigma \inv/{\1\T \bSigma \inv\1}$ and using the following equalities
	\begin{equation*}
	\frac{1}{\1\T \bSigma \inv\1}=\VGMV ,~~
	\frac{\1\T \bSigma \inv{\bmu}}{\1\T \bSigma \inv\1}=\RGMV ,~~ \text{and}~~
	\frac{{\bmu}\T \bSigma \inv{\bmu}}{\1\T \bSigma \inv\1}=\RGMV ^2+s \VGMV 
	\end{equation*}
	with $\RGMV $, $\VGMV $, and $s$ as in \eqref{eq: eff set const: R, V, s}, we get
		\begin{eqnarray}\label{eq:system_of_unknowns}
		&-(\gamma+1)\frac{\RGMV ^2+s \VGMV } {X} +\gamma\frac{X(\VGMV +\RGMV ^2+s \VGMV )}{Y}+\RGMV =0 , \\
		&-(\gamma+1)\frac{\RGMV }{X}+\gamma\frac{\VGMV + X \RGMV }{Y}+1=0 .\label{eq:system_of_unknowns2}	
		\end{eqnarray}
	Next, we multiply (\ref{eq:system_of_unknowns}) {by $\RGMV $, multiply (\ref{eq:system_of_unknowns2}) by $\RGMV ^2+s \VGMV  $, and subtract the first equation from the second one to get
		\begin{equation}\label{eq: new_system}	
		\gamma\frac{X \RGMV  -\RGMV ^2-s\VGMV }{Y}=s
		\end{equation}
		and, consequently,
		\begin{equation}\label{eq:y_solution}
		\gamma\frac{1}{Y}=\frac{s}{X \RGMV  -\RGMV ^2-s\VGMV }.
		\end{equation}
		Substituting \eqref{eq:y_solution} to \eqref{eq:system_of_unknowns2} leads to
		\begin{eqnarray*}\label{eq: new_system2}
			-(\gamma+1)\frac{\RGMV }{X}+\frac{s(\VGMV + X \RGMV )}{X \RGMV  -\RGMV ^2-s\VGMV }+1=0
		\end{eqnarray*}
		or equivalently to
		\begin{equation}\label{eq:X_final}
		(1+s)X^2-(\gamma+2)\RGMV  X+(\gamma+1)(\RGMV ^2+s\VGMV )=0
		\end{equation}
	}
	
	The roots of the quadratic equation \eqref{eq:X_final} are given by
	\begin{equation}\label{eq:x_1,2}
	X_{\pm}=\frac{(\gamma+2)\RGMV  \pm \sqrt{\D}}{2(1+s)},
	\end{equation}
	where
	\begin{equation*}
	\D=(\gamma+2)^2\RGMV ^2-4(\gamma+1)(1+s)(\RGMV ^2+s\VGMV ).
	\end{equation*}
	Finally, from \eqref{eq:y_solution} and \eqref{eq:x_1,2}, we get
	\begin{equation}\label{eq:y_solution1}
	Y_{\pm}=\frac{\gamma}{s}(X_{\pm} \RGMV  -\RGMV ^2-s\VGMV ).
	\end{equation}
	Moreover, equation (\ref{eq:x_1,2}) shows that the quadratic equation has a solution if and only if $\D \ge 0$ which is equivalent to \eqref{gamma_min}. Indeed, taking the derivative of $\D(\gamma)$ with respect to $\gamma$ and setting it equal to zero one gets
        \begin{eqnarray*}
          \D'(\gamma)=2(\gamma+2)R^2_{GMV}-4(1+s)(R^2_{GMV}+sV_{GMV})=0\,
        \end{eqnarray*}
        with the solution $\gamma^*=\frac{2(1+s)(R^2_{GMV}+sV_{GMV})}{R^2_{GMV}}-2=2s\left(1+(1+s)\frac{V_{GMV}}{R^2_{GMV}}\right)$. One can easily check that $\D(\gamma^*)$ is negative, thus, $\gamma^*$ can not be the minimum possible value of $\gamma$ which guarantees the existence of the solution of \eqref{eq:X_final}. On the other side, the second derivative of $\D(\gamma)$ is positive, which indicates that $\gamma^*$ is the global minimum point of the parabola $\D(\gamma)$. Now from $\D(0)<0$ follows that the smallest possible value of $\gamma>0$ should be the positive solution of equation $\D(\gamma)=0$, which is exactly $\gamma_{min}$ presented in \eqref{gamma_min}.
	
	Further, inserting two solutions of $X$ and $Y$ into (\ref{eq: lagrange derivetive}), we get two equations for the possible optimal portfolio weights. Since the method of Lagrange multipliers only provides a necessary condition for the optimality in constrained problems we have to analyze the critical points in more detail to check whether a maximum or a minimum is present. We consider the right side of (\ref{eq:normal_expectation}) for $(X_+,Y_+)$ and $(X_-,Y_-)$ and prove for all $\gamma \geq \gamma_{min}$ and $\gamma \neq 1$ that
	\begin{equation}\label{ieq: maximum}
	L(\gamma, X_-, Y_-) > L(\gamma, X_+, Y_+)
	\end{equation}
	 and, thus, the maximum is attained at $(X_-,Y_-)$.
	
	
In the case of $\gamma>1$, \eqref{ieq: maximum} is equivalent to
\begin{eqnarray*}
&&(1-\gamma^2)\ln X_{+} + \frac{1}{2}(\gamma^2-\gamma)\ln Y_{+}-(1-\gamma^2)\ln X_{-} + \frac{1}{2}(\gamma^2-\gamma)\ln Y_{-}\\
&=&(1-\gamma^2)\ln \frac{X_{+}}{X_{-}}+\frac{1}{2}(\gamma^2-\gamma)\ln\frac{Y_{+}}{Y_{-}}>0
\end{eqnarray*}
or

\vspace{-1cm}

\begin{eqnarray}\label{ieq:final}
(1+\gamma)\ln \frac{X_{-}}{X_{+}}+\frac{1}{2}\gamma\ln\frac{Y_{+}}{Y_{-}}>0
\end{eqnarray}

Similarly, for $0<\gamma<1$ we get
\begin{eqnarray*}
&&(1-\gamma^2)\ln X_{+} + \frac{1}{2}(\gamma^2-\gamma)\ln Y_{+}-(1-\gamma^2)\ln X_{-} + \frac{1}{2}(\gamma^2-\gamma)\ln Y_{-}\\
&=&(1-\gamma^2)\ln \frac{X_{+}}{X_{-}}+\frac{1}{2}(\gamma^2-\gamma)\ln\frac{Y_{+}}{Y_{-}}<0
\end{eqnarray*}
which coincides with \eqref{ieq:final}. 

Consequently, in order to prove \eqref{ieq: maximum}, it is sufficient to show \eqref{ieq:final} for $\gamma>\gamma_{min}$ such that $\gamma \neq1$. Or, equivalently, we have to show that for all $\gamma>\gamma_{min}$ the following inequality holds
	\begin{equation}\label{ineq: gamma>1 proof}	
	-\frac{\gamma+1}{\gamma}\ln\frac{X_{+}}{X_{-}}+\frac{1}{2}\ln\frac{Y_{+}}{Y_{-}} > 0.
	\end{equation}	
	First, we prove that the derivative of the function on the left side of (\ref{ineq: gamma>1 proof}) with respect to $\gamma$ is positive. It holds that
	
	\begin{equation*}
	\begin{split}
	&\frac{\partial}{\partial \gamma}\left(-\frac{\gamma+1}{\gamma}\ln\frac{X_{+}}{X_{-}}+\frac{1}{2}\ln\frac{Y_{+}}{Y_{-}}\right)=\frac{1}{\gamma^2}\ln\frac{X_+}{X_-}\\
	&-\frac{\gamma+1}{\gamma}\frac{X_-}{X_+}\frac{\partial}{\partial \gamma}\left(\frac{X_+}{X_-}\right)+\frac{1}{2}\frac{Y_-}{Y_+}\frac{\partial}{\partial \gamma}\left(\frac{Y_+}{Y_-}\right),
	\end{split}
	\end{equation*}
	Now we calculate $\frac{\partial}{\partial \gamma}\left(\frac{X_+}{X_-}\right)$ and $\frac{\partial}{\partial \gamma}\left(\frac{Y_+}{Y_-}\right)$. It holds
	\begin{eqnarray}
          &&\frac{\partial}{\partial\gamma}\left(\frac{X_+}{X_-}\right)=\frac{\partial}{\partial \gamma}\left(\frac{(\gamma+2) \RGMV  +\sqrt{\D}}{(\gamma+2) \RGMV -\sqrt{\D}}\right)\nonumber\\
	&=&\frac{\left( \RGMV +\frac{\partial}{\partial\gamma}\left(\sqrt{\D}\right)\right)\left((\gamma+2) \RGMV -\sqrt{\D}\right)}{\left((\gamma+2) \RGMV -\sqrt{\D}\right)^2}\nonumber\\
	&-&\frac{\left( \RGMV -\frac{\partial}{\partial\gamma}\left(\sqrt{\D}\right)\right)\left((\gamma+2) \RGMV +\sqrt{\D}\right)}{\left((\gamma+2) \RGMV -\sqrt{\D}\right)^2}\nonumber\\
	&=&2\frac{\left(\gamma+2\right) \RGMV \frac{\partial}{\partial\gamma}\left(\sqrt{\D}\right)- \RGMV \sqrt{\D}}{\left((\gamma+2) \RGMV -\sqrt{\D}\right)^2}\nonumber\\
	&=&2\frac{ \RGMV }{\sqrt{\D}}\frac{\left(\gamma+2\right)\frac{\frac{\partial}{\partial\gamma}\D}{2}-\D}{\left((\gamma+2) \RGMV -\sqrt{\D}\right)^2}=\frac{ \RGMV }{\sqrt{\D}}\frac{4\gamma(1+s)\left( \RGMV ^2+s\VGMV  \right)}{\left((\gamma+2) \RGMV -\sqrt{\D}\right)^2}\nonumber\\
	&=&\frac{\gamma}{\gamma+1}\frac{ \RGMV }{\sqrt{\D}}\frac{(\gamma+2) \RGMV +\sqrt{\D}}{(\gamma+2) \RGMV -\sqrt{\D}}=\frac{\gamma}{\gamma+1}\frac{ \RGMV }{\sqrt{\D}}\frac{X_+}{X_-}\label{Xplus_Xminus},
	\end{eqnarray}
	where second last equality follows from $(1+s)(R^2_{GMV}+sV_{GMV})=\frac{(\gamma+2)^2R^2_{GMV}-\D}{4(1+\gamma)}$.  In order to proceed to the derivative of $\frac{Y_+}{Y_-}$ we first note that
          \begin{eqnarray*}
            \frac{Y_+}{Y_-}&=&\frac{((\gamma+2) \RGMV +\sqrt{\D})R_{GMV}-2(1+s)(R^2_{GMV}+sV_{GMV})}{((\gamma+2) \RGMV -\sqrt{\D})R_{GMV}-2(1+s)(R^2_{GMV}+sV_{GMV})}\\
            &=&\frac{\left((\gamma+2) \RGMV +\sqrt{\D}\right)\left(\gamma \RGMV +\sqrt{\D}\right)}{\left((\gamma+2) \RGMV -\sqrt{\D}\right)\left(\gamma \RGMV -\sqrt{\D}\right)}=\frac{X_+\left(\gamma \RGMV +\sqrt{\D}\right)}{X_-\left(\gamma \RGMV -\sqrt{\D}\right)}\,.
          \end{eqnarray*}
      Thus, taking the derivative of $\frac{Y_+}{Y_-}$ leads to 
	\begin{equation*}
	\begin{split}
          &\frac{\partial}{\partial \gamma}\left(\frac{Y_+}{Y_-}\right)
          =\frac{\gamma \RGMV +\sqrt{\D}}{\gamma \RGMV -\sqrt{\D}}\frac{\partial}{\partial\gamma}\left( \frac{X_+}{X_-}\right)+\frac{X_+}{X_-}\frac{\partial}{\partial\gamma}\left(\frac{\gamma \RGMV +\sqrt{\D}}{\gamma \RGMV -\sqrt{\D}}\right)\\
          &= \frac{\gamma \RGMV +\sqrt{\D}}{\gamma \RGMV -\sqrt{\D}}\frac{\gamma}{\gamma+1}\frac{ \RGMV }{\sqrt{\D}}\frac{X_+}{X_-} +\frac{X_+}{X_-}\frac{\partial}{\partial\gamma}\left(\frac{\gamma \RGMV +\sqrt{\D}}{\gamma \RGMV -\sqrt{\D}}\right)\\
          &= \frac{\gamma}{\gamma+1}\frac{ \RGMV }{\sqrt{\D}}\frac{Y_+}{Y_-}+\frac{X_+}{X_-}\frac{\partial}{\partial\gamma}\left(\frac{\gamma \RGMV +\sqrt{\D}}{\gamma \RGMV -\sqrt{\D}}\right)\,.
	\end{split}
        	\end{equation*}
        Thus, we only need to calculate $\frac{\partial}{\partial\gamma}\left(\frac{\gamma \RGMV +\sqrt{\D}}{\gamma \RGMV -\sqrt{\D}}\right)$. It holds
	\begin{equation*}
	\begin{split}
	&\frac{\partial}{\partial \gamma}\left(\frac{\gamma \RGMV  +\sqrt{\D}}{\gamma \RGMV -\sqrt{\D}}\right)=\frac{\left( \RGMV +\frac{\partial}{\partial\gamma}\left(\sqrt{\D}\right)\right)\left(\gamma \RGMV -\sqrt{\D}\right)}{\left(\gamma \RGMV -\sqrt{\D}\right)^2}\\
	&-\frac{\left( \RGMV -\frac{\partial}{\partial\gamma}\left(\sqrt{\D}\right)\right)\left(\gamma \RGMV +\sqrt{\D}\right)}{\left(\gamma \RGMV -\sqrt{\D}\right)^2}\\
	&=2\frac{\gamma \RGMV \frac{\partial}{\partial\gamma}\left(\sqrt{\D}\right)- \RGMV \sqrt{\D}}{\left(\gamma \RGMV -\sqrt{\D}\right)^2}=2\frac{ \RGMV }{\sqrt{\D}}\frac{\gamma\frac{\frac{\partial}{\partial\gamma}\D}{2}-\D}{\left(\gamma \RGMV -\sqrt{\D}\right)^2}\\
	&=\frac{\gamma+2}{\gamma+1}\frac{\RGMV }{\sqrt{\D}}\frac{\gamma \RGMV +\sqrt{\D}}{\gamma \RGMV -\sqrt{\D}},
	\end{split}
      \end{equation*}
	which immediately implies that
	\begin{equation*}
	\begin{split}
	\frac{\partial}{\partial \gamma}\left(\frac{Y_+}{Y_-}\right)=\frac{\gamma}{\gamma+1}\frac{ \RGMV }{\sqrt{\D}}\frac{Y_+}{Y_-}+\frac{\gamma+2}{\gamma+1}\frac{\RGMV }{\sqrt{\D}}\frac{Y_+}{Y_-}=2\frac{\RGMV }{\sqrt{\D}}\frac{Y_+}{Y_-}.
	\end{split}
	\end{equation*}
	
	Putting all together it follows that for $\gamma > \gamma_{min}$
	\begin{equation*}
	\frac{\partial}{\partial \gamma}\left(-\frac{\gamma+1}{\gamma}\ln\frac{X_{+}}{X_{-}}+\frac{1}{2}\ln\frac{Y_{+}}{Y_{-}}\right)=\frac{1}{\gamma^2}\ln\frac{X_+}{X_-}>0.
	\end{equation*}
	
	

Thus, the function on the left side in \eqref{ineq: gamma>1 proof} is monotonically increasing function of $\gamma$ and, thus, attains its minimum value at $\gamma=\gamma_{min}$. So, calculating the value of function in \eqref{ineq: gamma>1 proof}  at point $\gamma=\gamma_{min}$ we get

	
	\begin{equation*}
	\ln\frac{X_{-}}{X_{+}}-\frac{\gamma}{2(\gamma+1)}\ln\frac{Y_{-}}{Y_{+}}\bigg|_{\gamma=\gamma_{min}}=\ln\frac{(\gamma+2) \RGMV  -\sqrt{\D}}{(\gamma+2) \RGMV +\sqrt{\D}}\bigg|_{\gamma=\gamma_{min}}=0,
	\end{equation*}
	because when $\gamma=\gamma_{min}$ it holds $\D=0$. Hence, the inequality \eqref{ineq: gamma>1 proof} holds for all $\gamma>\gamma_{min}$ and, thus, the maximum is attained at $\left(X_{-},Y_{-}\right)$	
\end{proof}

\begin{proof}[Proof of Corollary \ref{cor: parabola}]
It holds that
\begin{eqnarray*}
	\w^*&=&\bSigma \inv\left[\left(-\1+(\gamma+1)\frac{{\bmu} } {X}\right)\frac{Y}{\gamma}-X{\bmu} \right]\\
&=& \w_{GMV}-\left(\frac{Y}{\gamma}+\VGMV \right)\bSigma\inv\1+ \left(\frac{(\gamma+1)Y}{X\gamma}-X\right) \bSigma\inv{\bmu},
\end{eqnarray*}
where in using \eqref{eq:y_solution} we get
\begin{eqnarray*}
	\frac{Y}{\gamma}+\VGMV =\frac{X \RGMV  -\RGMV ^2-s\VGMV }{s} +\VGMV =\frac{X-\RGMV }{s}\RGMV 
\end{eqnarray*}
and
\begin{eqnarray*}
\frac{(\gamma+1)Y}{X\gamma}-X&=&\frac{\gamma+1}{X}\frac{X \RGMV  -\RGMV ^2-s\VGMV }{s}-X\\
&=&\frac{(\gamma+1)\RGMV X-(\gamma+1)(\RGMV ^2-s\VGMV ) -X^2 s}{Xs}\\
&=&\frac{X^2-\RGMV X}{Xs}=\frac{X-\RGMV }{s},
\end{eqnarray*}
where the third line follows from \eqref{eq:X_final}.

Hence,
\begin{eqnarray*}
\w^*&=& \w_{GMV}-\frac{X-\RGMV }{s}\RGMV \bSigma\inv\1+ \frac{X-\RGMV }{s} \bSigma\inv{\bmu}\\
&=&\w_{GMV}+\frac{X-\RGMV }{s}\boldsymbol{Q}{\bmu},
\end{eqnarray*}
\end{proof}

\begin{proof}[Proof of Theorem \ref{th: power_utility_mv-efficiency}]
  First part, i.e., $\gamma\geq\gamma_{min}$, follows obviously from the condition \eqref{gamma_min}.
  For the second one using \eqref{ineq: condition of theorem: optimal portfolio, power utility } we first observe that
\begin{eqnarray}
\D&=&(\gamma+2)^2\RGMV ^2-4(\gamma+1)(1+s)(\RGMV ^2+s\VGMV )\nonumber\\
&=&\left((\gamma+2)^2-4(\gamma+1)(1+s)\right)\RGMV ^2-4(\gamma+1)(1+s)s\VGMV  \nonumber\\
&=&\left((\gamma-2s)^2-4s(1+s)\right)\RGMV ^2-4(\gamma+1)(1+s)s\VGMV \nonumber\\
&=&(\gamma-2s)^2 \RGMV ^2-4(1+s)s(\RGMV ^2+(\gamma+1)\VGMV )\geq0\label{con_on_D}.
\end{eqnarray}

Next, in order to prove the second inequality in \eqref{con: power_utility_mv-efficiency}, we only have to find in which case
\begin{equation}\label{con}
X-\RGMV >0.
\end{equation}
Inserting $X$ from \eqref{eq:X} in the above equation \eqref{con} we get
\begin{eqnarray}
\frac{(\gamma-2s)\RGMV -\sqrt{\D}}{2(1+s)}> 0\label{con_on_R}.
\end{eqnarray}
So, because of \eqref{con_on_D} we have $(\gamma-2s)^2R^2_{GMV}\geq \D$ and together with $\gamma\geq\gamma_{min}>2s$ and \eqref{con_on_R} it implies that \eqref{con} is satisfied $\iff$ $\RGMV >0$.
\end{proof}

\begin{proof}[Proof of Corollary \ref{cor: Sharpe ratio if gamma=infinity}]
For the proof of the corollary, we only have to compute the following limit
\begin{eqnarray*}
\lim\limits_{\gamma\rightarrow\infty} ((\gamma+2)\RGMV  - \sqrt{\D})
&=&\lim\limits_{\gamma\rightarrow\infty} \frac{(\gamma+2)^2\RGMV ^2 - D}{(\gamma+2)\RGMV  + \sqrt{\D}}\\
&=&\lim\limits_{\gamma\rightarrow\infty} \frac{4(\gamma+1)(1+s)(\RGMV ^2+s\VGMV )}{(\gamma+2)\RGMV  + \sqrt{\D}}\\
&=&\frac{2(1+s)(\RGMV ^2+s\VGMV )}{\RGMV }\,.
\end{eqnarray*}

Now, the application of Corollary \ref{cor: parabola} leads to
\begin{eqnarray*}
\lim\limits_{\gamma\rightarrow\infty}\w^*&=& \w_{GMV}+\frac{\lim\limits_{\gamma\rightarrow\infty}X-\RGMV }{s}\boldsymbol{Q}{\bmu}\\
&=&\w_{GMV}+\frac{\RGMV ^2+s\VGMV -\RGMV ^2}{s\RGMV }\boldsymbol{Q}{\bmu}\\
&=&\w_{GMV}+\frac{\VGMV }{\RGMV }\boldsymbol{Q}{\bmu}
=\frac{\bSigma\inv{\bmu}}{\1\T \bSigma \inv{\bmu}}.
\end{eqnarray*}

\end{proof}

\begin{proof}[Proof of Corollary \ref{cor: Sharp greater}]
Using \eqref{eq:Y} and \eqref{eq:X_final}, the variance of the optimal portfolio that maximizes the expected power utility is given by
\begin{eqnarray*}
V &=& Y-X^2\\
&=& \frac{\gamma}{s}\left(X\RGMV -R^2_{GMV}-s\VGMV \right)\\
&-& \frac{\gamma+2}{s+1}X\RGMV +\frac{\gamma+1}{s+1}\left(R^2_{GMV}+s\VGMV \right) \\
&=& \frac{\gamma-2s}{s(1+s)}X\RGMV  -\frac{\gamma-s}{s(1+s)}\left(R^2_{GMV}+s\VGMV \right),
\end{eqnarray*}
which is non-negative if and only if
\begin{eqnarray*}
X&\ge& \frac{\gamma-s}{(\gamma-2s)\RGMV } \left(R^2_{GMV}+s\VGMV \right)\\
&>&\frac{1}{\RGMV } \left(R^2_{GMV}+s\VGMV \right)=\frac{\bmu\T \bSigma\inv\bmu}{\1\T \bSigma\inv\bmu}.
\end{eqnarray*}
\end{proof}

\begin{proof}[Proof of Corollary \ref{cor: Decrease of variance with the increase of gamma}]
	Using \eqref{eq: Mean-Variance parabola} the variance of the optimal portfolio that maximizes the expected power utility is given by
	
\begin{equation*}
V=\frac{1}{s}\left(X-\RGMV \right)^2+\VGMV .
\end{equation*}
	Next we show that a partial derivative of $X$ with respect to $\gamma$ is negative what brings the decrease of $X$ and $V$ by $\gamma$.
	
\begin{equation*}
\begin{split}
	&\frac{\partial X}{\partial\gamma}=\frac{1}{2(1+s)}\left(\RGMV -\frac{2(\gamma+2)\RGMV ^2-4(1+s)\left(\RGMV ^2+s\VGMV \right)}{2\sqrt{\D}}\right)\\
	&=\frac{4\D\RGMV ^2-\left[2(\gamma+2)\RGMV ^2-4(1+s)\left(\RGMV ^2+s\VGMV \right)\right]^2}{4(1+s)\sqrt{\D}\left[2\sqrt{\D}\RGMV +2(\gamma+2)\RGMV ^2-4(1+s)\left(\RGMV ^2+s\VGMV \right)\right]}\\
	&=\frac{16(1+s)\left(\RGMV ^2+s\VGMV \right)\RGMV ^2-16(1+s)^2\left(\RGMV ^2+s\VGMV \right)^2}{4(1+s)\sqrt{\D}\left[2\sqrt{\D}\RGMV +2(\gamma+2)\RGMV ^2-4(1+s)\left(\RGMV ^2+s\VGMV \right)\right]}\\
	&=\frac{4\left(\RGMV ^2+s\VGMV \right)\left[\RGMV ^2-(1+s)\left(\RGMV ^2+s\VGMV \right)\right]}{\sqrt{\D}\left[2\sqrt{\D}\RGMV +2(\gamma+2)\RGMV ^2-4(1+s)\left(\RGMV ^2+s\VGMV \right)\right]}<0.	
\end{split}
\end{equation*}	
	 
\end{proof}

\begin{proof}[Proof of Theorem \ref{th: optimal portfolio gamma=1}]
For the logarithmic utility function it holds with Lemma \ref{lem1} that
\begin{equation}\label{eq:normal_expectation_cor}
		E\left[U(W)\right]=\ln W_0+2\ln \w\T {\bmu} -\frac{1}{2}\ln\left(\w\T \bSigma \w+(\w\T {\bmu} )^2\right).
\end{equation}

In order to maximize the expected logarithmic utility \eqref{eq:normal_expectation_cor} under the constrain $\w\T \1=1$, the method of Lagrange multipliers is used with the Lagrange function given by
\begin{equation}\label{eq:maximization_problem_cor}
		{\ln W_0+}2\ln \w\T {\bmu} -\frac{1}{2}\ln\left(\w\T \bSigma \w+(\w\T {\bmu} )^2\right) + \lambda \,  (\w\T \1-1) .
\end{equation}		
		
Partial derivation yields
		\begin{align}\label{eq: lagrange derivetive_cor}
		&\frac{\partial}{\partial \w}L=2\frac{{\bmu}} {\w\T {\bmu} }-\frac{\bSigma \w+(\w\T {\bmu} ){\bmu} }{\w\T \bSigma \w+(\w\T {\bmu} )^2}+\lambda\1=\mathbf{0} ,
		\\
		&\frac{\partial}{\partial \lambda}L=\w \T \1-1=0	.
		\end{align}

The multiplication of \eqref{eq: lagrange derivetive_cor} by $\w\T $ leads to $\lambda = -1$. {Using the definition of $X$ and $Y$ from \eqref{XY} and multiplying \eqref{eq: lagrange derivetive_cor} by ${\bmu} \T \bSigma \inv/{\1\T \bSigma \inv\1}$ and $\1\T \bSigma \inv/{\1\T \bSigma \inv\1}$, we get
\begin{eqnarray}\label{eq:system_of_unknowns_cor}
	&2\frac{\RGMV ^2+s \VGMV } {X} -\frac{X(\VGMV +\RGMV ^2+s \VGMV )}{Y}-\RGMV =0 , \\
	&2\frac{\RGMV }{X}-\frac{\VGMV + X \RGMV }{Y}-1=0 .\label{eq:system_of_unknowns2_cor}	
\end{eqnarray}

The system of equations \eqref{eq:system_of_unknowns_cor} and \eqref{eq:system_of_unknowns2_cor} is a partial case of \eqref{eq:system_of_unknowns} and \eqref{eq:system_of_unknowns2} given in the proof of Theorem \ref{th: optimal portfolio weights theorem} which corresponds to $\gamma=1$. Hence, its solution is given by
\begin{equation}\label{eq:x_1,2y1,2_cor}
	X_{\pm}=\frac{3\RGMV  \pm \sqrt{\D}}{2(1+s)} \quad \text{and} \quad Y_{\pm}=\frac{X_{\pm} \RGMV  -\RGMV ^2}{s} -\VGMV ,
	\end{equation}
	where
	\begin{equation*}
	\D=9\RGMV ^2-8(1+s)(\RGMV ^2+s\VGMV ).
\end{equation*}

Finally, for $\D > 0$ or, equivalently, for $\gamma_{min}\leq1$  we get that
\begin{eqnarray*}
L(X_-, Y_-) - L(X_+, Y_+)&=&2\ln X_{-} -\frac{1}{2}\ln Y_{-}-2\ln X_{+} +\frac{1}{2}\ln Y_{+}\\
&=&2\ln \frac{X_{-}}{X_{+}}+\frac{1}{2}\ln \frac{Y_{+}}{Y_{-}}>0.
\end{eqnarray*}
Hence, $(X_{-},Y_{-})$ maximizes the expected logarithmic utility.
}
\end{proof}

{
\begin{proof}[Proof of Theorem \ref{th: log_utility_mv-efficiency}]
The proof of Theorem \ref{th: log_utility_mv-efficiency} follows from the proof of Theorem \ref{th: power_utility_mv-efficiency} with $\gamma=1$.
\end{proof}
}
{\small
\bibliography{mybib_single_period}{}

\begin{thebibliography}{}

\bibitem[\protect\astroncite{Bachelier}{1900}]{bachelier1900theorie}
Bachelier, L. (1900).
\newblock {\em Th{\'e}orie de la sp{\'e}culation}.
\newblock Gauthier-Villars.

\bibitem[\protect\astroncite{Barberis}{2000}]{barberis2000investing}
Barberis, N. (2000).
\newblock Investing for the long run when returns are predictable.
\newblock {\em The Journal of Finance}, 55(1):225--264.

\bibitem[\protect\astroncite{Bodnar et~al.}{2013}]{bodnar2013equivalence}
Bodnar, T., Parolya, N., and Schmid, W. (2013).
\newblock On the equivalence of quadratic optimization problems commonly used
  in portfolio theory.
\newblock {\em European Journal of Operational Research}, 229(3):637--644.

\bibitem[\protect\astroncite{Bodnar et~al.}{2015a}]{bodnar2015closed}
Bodnar, T., Parolya, N., and Schmid, W. (2015a).
\newblock A closed-form solution of the multi-period portfolio choice problem
  for a quadratic utility function.
\newblock {\em Annals of Operations Research}, 229(1):121--158.

\bibitem[\protect\astroncite{Bodnar et~al.}{2015b}]{bodnar2015exact}
Bodnar, T., Parolya, N., and Schmid, W. (2015b).
\newblock On the exact solution of the multi-period portfolio choice problem
  for an exponential utility under return predictability.
\newblock {\em European Journal of Operational Research}, 246(2):528--542.

\bibitem[\protect\astroncite{Bodnar and Schmid}{2009}]{bodnar2009econometrical}
Bodnar, T. and Schmid, W. (2009).
\newblock Econometrical analysis of the sample efficient frontier.
\newblock {\em The European journal of finance}, 15(3):317--335.

\bibitem[\protect\astroncite{Brandt}{2009}]{brandt2009portfolio}
Brandt, M. (2009).
\newblock Portfolio choice problems.
\newblock {\em Handbook of financial econometrics}, 1:269--336.

\bibitem[\protect\astroncite{Brandt et~al.}{2005}]{brandt2005simulation}
Brandt, M.~W., Goyal, A., Santa-Clara, P., and Stroud, J.~R. (2005).
\newblock A simulation approach to dynamic portfolio choice with an application
  to learning about return predictability.
\newblock {\em The Review of Financial Studies}, 18(3):831--873.

\bibitem[\protect\astroncite{Brandt and Santa-Clara}{2006}]{brandt2006dynamic}
Brandt, M.~W. and Santa-Clara, P. (2006).
\newblock Dynamic portfolio selection by augmenting the asset space.
\newblock {\em The Journal of Finance}, 61(5):2187--2217.

\bibitem[\protect\astroncite{Campbell and
  Viceira}{2002}]{campbell2002strategic}
Campbell, J.~Y. and Viceira, L.~M. (2002).
\newblock {\em Strategic asset allocation: portfolio choice for long-term
  investors}.
\newblock Clarendon Lectures in Economic.

\bibitem[\protect\astroncite{Chen et~al.}{2011}]{chen2011modeling}
Chen, A., Pelsser, A., and Vellekoop, M. (2011).
\newblock Modeling non-monotone risk aversion using sahara utility functions.
\newblock {\em Journal of Economic Theory}, 146(5):2075--2092.

\bibitem[\protect\astroncite{Cover}{1991}]{cover1991universal}
Cover, T.~M. (1991).
\newblock Universal portfolios.
\newblock {\em Mathematical finance}, 1(1):1--29.

\bibitem[\protect\astroncite{Dybvig and Ingersoll}{1982}]{dybvig1982mean}
Dybvig, P.~H. and Ingersoll, J.~E. (1982).
\newblock Mean-variance theory in complete markets.
\newblock {\em Journal of Business}, pages 233--251.

\bibitem[\protect\astroncite{Elton and Gruber}{1974}]{elton1974portfolio}
Elton, E.~J. and Gruber, M.~J. (1974).
\newblock Portfolio theory when investment relatives are lognormally
  distributed.
\newblock {\em The Journal of Finance}, 29(4):1265--1273.

\bibitem[\protect\astroncite{Elton et~al.}{2009}]{elton2009modern}
Elton, E.~J., Gruber, M.~J., Brown, S.~J., and Goetzmann, W.~N. (2009).
\newblock {\em Modern portfolio theory and investment analysis}.
\newblock John Wiley \& Sons.

\bibitem[\protect\astroncite{Fama}{1965}]{fama1965portfolio}
Fama, E.~F. (1965).
\newblock Portfolio analysis in a stable paretian market.
\newblock {\em Management science}, 11(3):404--419.

\bibitem[\protect\astroncite{Grauer}{1986}]{grauer_1986}
Grauer, R.~R. (1986).
\newblock Normality, solvency, and portfolio choice.
\newblock {\em Journal of Financial and Quantitative Analysis}, 21(3):265--278.

\bibitem[\protect\astroncite{Grauer and Hakansson}{1987}]{grauer1987gains}
Grauer, R.~R. and Hakansson, N.~H. (1987).
\newblock Gains from international diversification: 1968--85 returns on
  portfolios of stocks and bonds.
\newblock {\em The Journal of Finance}, 42(3):721--739.

\bibitem[\protect\astroncite{Hakansson}{1971}]{hakansson1971multi}
Hakansson, N.~H. (1971).
\newblock Multi-period mean-variance analysis: Toward a general theory of
  portfolio choice.
\newblock {\em The Journal of Finance}, 26(4):857--884.

\bibitem[\protect\astroncite{Jondeau et~al.}{2007}]{jondeau2007financial}
Jondeau, E., Poon, S.-H., and Rockinger, M. (2007).
\newblock {\em Financial modeling under non-Gaussian distributions}.
\newblock Springer Science \& Business Media.

\bibitem[\protect\astroncite{Levy and Markowitz}{1979}]{levy1979approximating}
Levy, H. and Markowitz, H.~M. (1979).
\newblock Approximating expected utility by a function of mean and variance.
\newblock {\em The American Economic Review}, 69(3):308--317.

\bibitem[\protect\astroncite{Limpert et~al.}{2001}]{LimpertStahel2001}
Limpert, E., Stahel, W.~A., and Abbt, M. (2001).
\newblock {L}og-normal {D}istributions across the {S}ciences: {K}eys and
  {C}lues.
\newblock {\em BioScience}, 51(5):341--352.

\bibitem[\protect\astroncite{Lynch and Tan}{2010}]{lynchtan2010}
Lynch, A.~W. and Tan, S. (2010).
\newblock Multiple risky assets, transaction costs, and return predictability:
  Allocation rules and implications for u.s. investors.
\newblock {\em Journal of Financial and Quantitative Analysis},
  45(4):1015--1053.

\bibitem[\protect\astroncite{Markowitz}{1952}]{markowitz1952portfolio}
Markowitz, H. (1952).
\newblock Portfolio selection.
\newblock {\em The journal of finance}, 7(1):77--91.

\bibitem[\protect\astroncite{Markowitz}{2014}]{markowitz2014mean}
Markowitz, H. (2014).
\newblock Mean--variance approximations to expected utility.
\newblock {\em European Journal of Operational Research}, 234(2):346--355.

\bibitem[\protect\astroncite{McDonald}{1996}]{mcdonald1996}
McDonald, J.~B. (1996).
\newblock {P}robability distributions for financial models.
\newblock In {\em Statistical Methods in Finance}, volume~14 of {\em Handbook
  of Statistics}, pages 427 -- 461. Elsevier.

\bibitem[\protect\astroncite{Merton and Samuelson}{1974}]{merton1974fallacy}
Merton, R.~C. and Samuelson, P.~A. (1974).
\newblock Fallacy of the log-normal approximation to optimal portfolio
  decision-making over many periods.
\newblock {\em Journal of Financial Economics}, 1(1):67--94.

\bibitem[\protect\astroncite{Mittnik and Rachev}{1993}]{mittnik1993modeling}
Mittnik, S. and Rachev, S.~T. (1993).
\newblock Modeling asset returns with alternative stable distributions.
\newblock {\em Econometric reviews}, 12(3):261--330.

\bibitem[\protect\astroncite{Pennacchi}{2008}]{pennacchi2008theory}
Pennacchi, G.~G. (2008).
\newblock {\em Theory of asset pricing}.
\newblock Pearson/Addison-Wesley Boston.

\end{thebibliography}
}
\end{document}